\documentclass{article}

\usepackage{microtype}
\usepackage{graphicx}
\usepackage{subfigure}
\usepackage{booktabs} % for professional tables

\usepackage{hyperref}

\usepackage[accepted]{icml2021}

\icmltitlerunning{Exponentially Many Local Minima in Quantum Neural Networks}

\def\imghome{./figures}
\usepackage{amsfonts, amsmath, amssymb, amsthm, bm, bbm}
\usepackage{mathtools}
\usepackage{braket}
\usepackage{stmaryrd} %St Mary’s Road symbol font
\usepackage{xcolor}
\usepackage[utf8]{inputenc} % allow utf-8 input
\usepackage[T1]{fontenc}    % use 8-bit T1 fonts
\usepackage{hyperref}       % hyperlinks
\usepackage{url}            % simple URL typesetting
\usepackage{nicefrac}       % compact symbols for 1/2, etc.
\usepackage{thmtools, thm-restate}
\usepackage{xstring}

\theoremstyle{plain}
\newtheorem{theorem}{Theorem}
\declaretheorem[name=Lemma,sibling=theorem]{lemma}
\newtheorem{corollary}[theorem]{Corollary}

\theoremstyle{definition}
\newtheorem{definition}[theorem]{Definition}

\newtheorem{example}{Example}
\newtheorem{claim}{Claim}[section]

\def\01{\{0,1\}}

\newcommand{\tr}{\operatorname{tr}}

\def\({\left(}
\def\){\right)}

\def\complex{\mathbb{C}}
\def\real{\mathbb{R}}

\def\integer{\mathbb{Z}}

\def\<{\langle}
\def\>{\rangle}

\def\N{\mathcal{N}}
\def\X{\mathcal{X}}
\def\Y{\mathcal{Y}}

\def\B{\mathcal{B}}

\def\S{\mathcal{S}}

\def\1{\mathbbm{1}}
\def\EXP{\mathbb{E}}

\newcommand{\supref}[1]{S.M. Sect.~\ref{#1}}
\newcommand{\clanns}{ClaNNs}

\newcommand{\mlvec}[1]{\boldsymbol{\mathbf{#1}}}

\hypersetup{pdflang=en-US,
    pdftitle={Exponentially Many Local Minima in Quantum Neural Networks},
    pdfauthor={Xuchen You, Xiaodi Wu},
    pdfsubject={Proceedings of the International Conference on Machine Learning 2021}}

\begin{document}

\twocolumn[
\icmltitle{Exponentially Many Local Minima in Quantum Neural Networks}



\icmlsetsymbol{equal}{*}

\begin{icmlauthorlist}
\icmlauthor{Xuchen You}{QUICS,CSUMIACS}
\icmlauthor{Xiaodi Wu}{QUICS,CSUMIACS}
\end{icmlauthorlist}

\icmlaffiliation{QUICS}{Joint Center for Quantum Information and Computer
  Science, University of Maryland}
\icmlaffiliation{CSUMIACS}{Department of Computer Science and Institute for
  Advanced Computer Studies, University of Maryland}

\icmlcorrespondingauthor{X.You}{xyou@umd.edu}
\icmlcorrespondingauthor{X.Wu}{xwu@cs.umd.edu}

\icmlkeywords{Machine Learning, ICML, Quantum Machine Learning, Optimization}

\vskip 0.3in
]



\printAffiliationsAndNotice{}  % leave blank if no need to mention equal contribution

\begin{abstract}
Quantum Neural Networks (QNNs), or the so-called variational quantum circuits,
are important quantum applications both because of their similar promises as
classical neural networks and because of the feasibility of their implementation
on near-term intermediate-size noisy quantum machines (NISQ). However, the
training task of QNNs is challenging and much less understood. We conduct a
quantitative investigation on the landscape of loss functions of QNNs and
identify a class of simple yet extremely hard  QNN instances for training.
Specifically, we show for typical under-parameterized QNNs,
there exists a dataset that induces a loss function with the number of spurious
local minima depending exponentially on the number of parameters.
Moreover, we show the optimality of our construction by providing an almost
matching upper bound on such dependence.
While local minima in classical neural networks are due to non-linear
activations, in quantum neural networks local minima appear as a result of the
quantum interference phenomenon.
Finally, we empirically confirm that our constructions can
indeed be hard instances in practice with typical gradient-based optimizers, which
demonstrates the practical value of our findings. 
\end{abstract}

%%%%%%%%%%%%%%%%%%%%%%%%%%%%%%%%%%%%%%%%%%%%%%%%%%%%%%%%%%%%%%%%%%%%%%%%%%%%%%%%%%%%%%%%%%%

\section{Introduction}
\label{sec:intro}
\textbf{Motivations.} With the recent establishment of quantum supremacy~\cite{google-supremacy, Zhong1460}, the research of quantum computing has entered a new stage where near-term Noisy Intermediate-Scale Quantum (NISQ) computers~\cite{Preskill2018NISQ}
become the important platform for demonstrating quantum applications. \emph{Quantum Neural Networks (QNNs)} (e.g.,~\citet{farhi2018classification, farhi2014quantum}), or the so-called variational quantum method (e.g.,~\citet{NC-VQE}), are the major candidates of applications that can be implemented on NISQ machines.

Typical QNNs replace classical neural networks (\clanns{}), which are just
parameterized classical circuits, by quantum circuits with \emph{classically
  parameterized unitary gates}. Instead of a classical mapping in \clanns{}
~from input to output, QNNs use a \emph{quantum} one which could be very hard
for classical computation to simulate (e.g.,~\citet{harrow2017quantum}) and
hence provide potential quantum speedups for machine learning tasks (e.g., see
the survey by~\citet{biamonte2017quantum} and by~\citet{harrow2017quantum} and
examples in~\citet{schuld2019quantum} and in~\citet{havlivcek2019supervised}).
Moreover, given their quantum-mechanical nature, QNNs (or the variational quantum method) have also demonstrated huge promises in attacking problems in quantum chemistry and material science. 
Contrary to quantum supremacy tasks which serve only as a way to separate quantum and classical computational power but are not necessarily useful,  Google has recently used the same machine to demonstrate the variational quantum method in calculating accurate electronic structures -- an important task in quantum chemistry~\cite{google-VQE}. 
Please see the survey~\cite{benedetti2019parameterized} for more recent exciting developments of QNNs. 

Similar to the classical case, the success of QNN applications will critically depend on the effectiveness of the training procedure which optimizes a \emph{loss function}
in terms of the \emph{read-outs} and the \emph{parameters} of QNNs for specific applications. 
The design of effective training methods has been under intense investigation both empirically and theoretically for \clanns{}.
Moreover, understanding the landscape of the loss functions and designing
corresponding training/optimization methods have recently emerged as a
principled approach to tackle this problem:
\cite{auer1996exponentially,
  safran2017spurious, yun2018small, ding2019sub, venturi2018spurious} showed the existence of
spurious local minima for \clanns{}; In turn, \cite{kawaguchi2016deep,
  du2018power, soudry2016no, nguyen2017loss, li2018over} characterized conditions for benign landscapes in
terms of choice of activation, loss function and (over)-parameterization,
providing insights on the design of \clanns{} and motivating explanations to
the success of gradient descent in
training \clanns{} in certain scenarios
\cite{jacot2018neural, arora2019exact, du2019gradient}; And training methods beyond simple
variants of gradient descent have been devised for training with
guarantees \cite{goel2017reliably, goel2019learning,
  zhong2017recovery, du2018improved}.

Much less has been understood for QNNs.
Most of the study of QNNs takes a trial-and-error approach by empirically comparing the performance of standard classical optimizers on training QNNs' loss functions~\cite{benedetti2019parameterized}.
It has been observed empirically that training QNNs could be very challenging due to the \emph{non-convex} nature of the corresponding loss functions (e.g.,~\citet{PhysRevA.101.012320, PhysRevA.97.022304}). 
However, these empirical studies are unfortunately restricted to small cases due to the limited access to quantum machines of reasonable sizes and the exponential cost in simulating them classically. 

A theoretical study on the training of QNNs would be more \emph{favorable} and \emph{scalable} given the limit on empirical study.
Indeed, a handful of such theoretical progress has been made. One prominent result is that random initialization of parameters will lead to vanishing gradients for much smaller size QNNs than \clanns{}~\cite{mcclean2018barren} and hence pose one unique training difficulty for QNNs. 
Most of the remaining theoretical results are about special cases of QNNs such as \emph{quantum approximate optimization algorithms} (QAOA) (e.g., ~\citet{farhi2014quantum, 2019arXiv191008187F}) and extremely over-parameterized cases (e.g.,~\citet{rabitz2004quantum,russell2016quantum,kiani2020learning}).

In this paper, we conduct a quantitative investigation on the landscape of loss functions for QNNs as a way to study their training issue. 
In particular, we are interested in understanding the properties of local minima of loss functions, such as, (1) the number of local minima depending on the architecture of QNNs, and (2) whether these local minima are \emph{benign} or \emph{spurious} ones, meaning that they are either close to the global minima or saddle points that can be escaped,
or they are truly bad local minima that will hinder the training procedure.
We are also motivated by the observation that QNNs share some similarity with linear neural networks without non-linear activation layers~\cite{kawaguchi2016deep} or one-hidden layer neural networks with quadratic activation~\cite{du2018power} that are both known to have only benign local minima.
The similarity is due to the fact that quantum mechanics underlying QNNs has a linear algebraic formulation similar to the linear part of \clanns{}. (Details in Section~\ref{sec:prelim}.)  
It is hence natural to wonder whether the local minima of QNNs could share these nice properties.

\textbf{Contributions. } Contrary to our original hope, we turn out to identify a class of \emph{simple yet extremely hard} instances of QNNs for the training.
Despite the similarity between QNNs and linear classical neural networks, we demonstrate that \emph{spurious} (or \emph{sub-optimal}) local minima do appear in QNNs and provide a quantitative characterization of the possible number of them. 
We focus on QNNs with the commonly used \emph{square loss} function under a
practical range of the number of parameters (or gates). Specifically, we
identify a general condition of under-parameterized QNNs, which we refer to as
QNNs \emph{with linear independence}. We show for such QNNs, a dataset can be
constructed such that the number of spurious local minima scales
\emph{exponentially} with the number of parameters.   
It demonstrates that QNNs behave quite differently from linear neural networks
(e.g.,~\citet{kawaguchi2016deep}) but share the feature of neurons with
\emph{non-linear} activation functions (e.g.,~\citet{auer1996exponentially}).  
This conceptual paradox could be explained by 
one central phenomenon of quantum mechanics behind QNNs called \emph{interference}. 
We observe that interference replaces the role of non-linear activation in
creating bad local minima for QNNs.
(Section~\ref{sec:construction})

We investigate further and prove that typical under-parameterized QNNs are indeed \emph{with linear independence}.
This indicates that for almost all under-parameterized QNNs, there is a dataset where training with
simple variants of gradient-based methods is hard.
(Section~\ref{sec:linear_indep})

Moreover, we show our construction is almost \emph{optimal} in terms of the dependence of the number of local minima on the number of parameters, by developing an almost matching upper bound. 
This upper bound also demonstrates a sharp separation between QNNs and \clanns{}: For \clanns{}, provided an arbitrary number of training samples, the number of local minima could be unbounded, and hence won't be upper bounded by any function of the number of parameters \cite{auer1996exponentially}. (Section~\ref{sec:upperbound})

Finally, we perform numerical experiments on concrete QNN instances with typical
optimizers, and empirically confirm that our constructions can indeed be hard instances in practice. 
These experiments strengthen the value of our theoretical findings on the practical end. (Section~\ref{sec:experiments})

It is worthwhile mentioning that our investigation on the landscape of loss
functions has a direct implication on the hardness of gradient-based methods.
While it does not rule out the possibility of efficient non-gradient-based
training, there are no obvious solutions to the efficient training for our constructions.
Identifying such training methods would be very interesting. 

\textbf{Related work.} There are only a few previous studies on the training of QNNs, each of which has targeted at some specific parameter range for QNNs.
The observation of vanishing gradients for random initialization of QNNs~\cite{mcclean2018barren} provides hard QNN instances for training, which, however, still require many layers to demonstrate the difficulty of training in practice.  
Our constructions are based on a general condition which includes simple special cases like 1-layer QNNs that are already able to demonstrate QNNs' training difficulty. 

Another line of work~\cite{rabitz2004quantum,russell2016quantum,kiani2020learning} considers the extremely over-parameterized QNN cases.
Specifically, when the number of parameters is comparable to the dimension of the underlying quantum system and the quantum \emph{controllability} condition can be established, all local minima of QNNs' loss functions will become global~\cite{rabitz2004quantum, russell2016quantum}. 
This theoretical prediction has also been observed empirically~\cite{kiani2020learning}. 
However, as the dimension of quantum systems grows \emph{exponentially} with the number of qubits, 
this over-parameterized case can hardly be realistic for any QNN of reasonable size.

\section{Preliminaries}
\label{sec:prelim}
\textbf{Supervised learning.} The goal of supervised learning is to identify a mapping from the feature space $\X$
to the label space $\Y$, given a training set $\{(\mlvec{x}_i,y_i)\}_{i=1}^m\subset (\X\times \Y)^m$ of $m$ samples of \emph{feature vectors} and \emph{labels}. A common practice to find a mapping based on a training set is through empirical risk minimization (ERM), finding a mapping that best align with the training sample with respect to a specific loss function $l:\Y\times\Y\rightarrow\real$. Let $\hat{y}_i$ be the prediction of a certain mapping given $\mlvec{x}_i$. The goal of ERM is to find the mapping that minimizes 
the average loss $\frac{1}{m}\sum_{i=1}^m l(\hat{y}_i, y_i)$. Throughout this paper we will consider square loss $l(\hat{y},y) = (\hat{y} - y)^2$.  

%%%%%%%%%%%%%%%%%%%%%%%%%%%%%%%%%%%%%%%%%%%%%%%%%%%%%%%%%%%%%%%%%%%%%%%%%%%%%%%%
%%%%%%%%%%%%%%%%%%%%%%%%%%%%%%%%%%%%%%%%%%%%%%%%%%%%%%%%%%%%%%%%%%%%%%%%%%%%%%%%

\textbf{Classical neural networks (\clanns{}).} Neural networks are parameterized families of mappings, widely considered for practical problems. Typical feed-forward neural networks are parameterized by a sequence of matrices $\{\mlvec{W}_i\}_{i=1}^t$, such that $\mlvec{W}_i\in\real^{d_i\times d_{i-1}}$, with $d_t = 1$ and $d_0$ is the same as the dimension of the feature space $\X$. For feature vector $\mlvec{x}$, the output $\hat{y}$ of the neural network is 
\begin{align}
    \hat{y} = \mlvec{W}_t\sigma(\mlvec{W}_{t-1}\sigma(\cdots\sigma(\mlvec{W}_1\mlvec{x})\cdots)),
\end{align}
where $\sigma(\cdot)$ denotes an 
element-wise activation
on the output of each layer. (See Figure~\ref{fig:neural-network}.) Linear neural networks~\cite{kawaguchi2016deep} is one special example where ~$\sigma$ is the identity mapping $\sigma(w) = w$:
$\hat{y} = \mlvec{W}_t\mlvec{W}_{t-1}\cdots\mlvec{W}_1\mlvec{x}$.
Another example is one-hidden layer neural networks with quadratic activation $\sigma(w) = w^2$~\cite{du2018power}, where the output 
$\hat{y} =\mlvec{x}^T\mlvec{W}_1^T\mlvec{W}_1\mlvec{x}$. Given the training set $\{(\mlvec{x}_i, y_i)\}_{i=1}^m$, the empirical risk minimization with square loss solves the optimization problem: 
\begin{align} \label{eqn:class}
\min_{\mlvec W_1}\frac{1}{m}\sum_{i=1}^m \bigl(\mlvec{x}_i^T\mlvec{W}_1^T\mlvec{W}_1\mlvec{x}_i-y_i\bigr)^2
\end{align}
A common choice of $\sigma(\cdot)$ is non-linear activation such as Relu or Sigmoid. These activations introduce \emph{non-linearity} which is the source of spurious local minima in neural networks \cite{kawaguchi2016deep,auer1996exponentially}.

%%%%%%%%%%%%%%%%%%%%%%%%%%%%%%%%%%%%%%%%%%%%%%%%%%%%%%%%%%%%%%%%%%%%%%%%%%%%%%%%
%%%%%%%%%%%%%%%%%%%%%%%%%%%%%%%%%%%%%%%%%%%%%%%%%%%%%%%%%%%%%%%%%%%%%%%%%%%%%%%%

\textbf{Quantum neural networks.} QNNs share the layered structure (Figure~\ref{fig:neural-network}) where a linear transformation $\mlvec{U}_i$ is applied on the output of the previous layer, however, with the following differences:

(1) \emph{Input}. The inputs to \clanns{} are feature vectors. Yet for QNNs, a feature vector $\mlvec{x}$ is first encoded into a quantum state $\mlvec{\rho}_{\mlvec{x}}$ then fed to the quantum circuits.
We are not restricted to specific encoding scheme (e.g.,  \cite{mitarai2018quantum,benedetti2019parameterized,lloyd2020quantum}).
For technical convenience, we will directly work with a set of $m$ samples of \emph{quantum encoding} and \emph{labels} 
 $ \S = \{(\mlvec\rho_i, y_i)\}_{i=1}^m $ where $\mlvec\rho_i$ encodes the information of $\mlvec x_i$. 

(2) \emph{Linear Transformation \& Parameterization}. The linear transformations $\{\mlvec{W}_i\}_{i=1}^t$ in \clanns{} could be general matrices, whereas the corresponding $\{\mlvec{U}_i\}_{i=1}^t$ in QNNs must be unitaries. 
Moreover, although $\{\mlvec{U}_i\}_{i=1}^t$ can be efficiently implemented by quantum machines, their classical representations are matrices of exponential dimension in terms of the system size (e.g., the number of qubits in QNNs). 
This makes classical simulation of QNNs extremely expensive and also makes the parameterizations of $\{\mlvec{U}_i\}_{i=1}^t$ different from the straightforward parameterizations of $\{\mlvec{W}_i\}_{i=1}^t$ (explained below).

(3) \emph{Output.} Contrary to \clanns{}, one needs to make a quantum
\emph{measurement} to read information from QNNs (explained below). While there
exist more advanced models of QNNs with additional nonlinearity, we consider
the most basic QNNs, where the measurements are the only source of slight non-linearity allowed by
quantum mechanics, which as we will see won't necessarily create bad local
minima for the training. Note further there is no direct counterpart of
classical non-linear activation $\sigma(\cdot)$ in QNNs of our consideration. 

\begin{figure}[!htbp]
  \centering
  \subfigure[Classical neural networks]{
    \includegraphics[width=.7\linewidth]{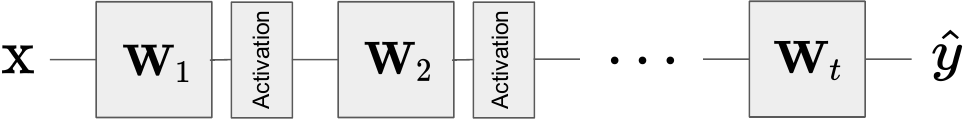}
  }
  \\
  \subfigure[Quantum neural networks]{
    \includegraphics[width=.83\linewidth]{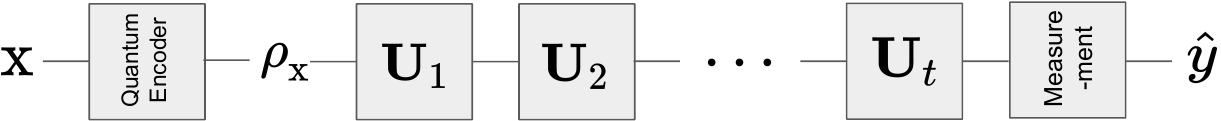}
  }
  \caption{
    An illustration of layer-structured classical and quantum neural networks.
    }
  \label{fig:neural-network}
\end{figure}

%%%%%%%%%%%%%%%%%%%%%%%%%%%%%%%%%%%%%%%%%%%%%%%%%%%%%%%%%%%%%%%%%%%%%%%%%%%%%%%%
%%%%%%%%%%%%%%%%%%%%%%%%%%%%%%%%%%%%%%%%%%%%%%%%%%%%%%%%%%%%%%%%%%%%%%%%%%%%%%%%

\textbf{Mathematical formulation of quantum
states.}
A general quantum state with dimension $d$ can be represented by a \emph{density operator} that is a positive
semidefinite (PSD) Hermitian matrix $\mlvec\rho\in\complex^{d\times d}$ with
$\tr(\mlvec\rho) = 1$.
A quantum state $\mlvec\rho$ is \emph{pure} if $\mlvec\rho = \mlvec v\mlvec v^\dagger$ for a $\ell_2$ unit vector $\mlvec v$.
A two-dimensional quantum state $\mlvec{\rho} \in \complex^{2\times 2}$ is usually referred as a \textit{qubit}, the quantum generalization of the classical binary bit.
The state of $n$ qubits lies in $\otimes_{i=1}^n\complex^{2\times 2}$ following the tensor product of spaces for single qubits, and is a linear operator on a Hilbert space with dimension $d=2^n$, i.e., scales \emph{exponentially} with the number of qubits $n$.

\textbf{Parameterization of quantum transformations.} Instead of directly parameterized matrices $W_i$, QNNs typically consist of \emph{classically parameterized} quantum gates. 
A general form of these gates is $\exp({-i\theta \mlvec H})$, where $\theta$ is
the parameter, $\mlvec H$ the Hamiltonian (i.e., a Hermitian matrix)
, and the exponential is a \emph{matrix} exponential.
For example, a commonly used gate set, called the Pauli rotation gate (e.g.,~\citet{farhi2018classification,li2017hybrid,ostaszewski2019qcsl}), can be expressed as $\exp({-i\theta\mlvec P_c})$ (on $c$-th qubit) or $\exp({-i\theta \mlvec{P}_c\otimes \mlvec{P}_{c^\prime}})$ (on $c$-th and $c'$-th qubits), where $\mlvec P_c$ refers to Pauli $\mlvec X, \mlvec Y, \mlvec Z $ matrices.\footnote{
    $
   \mlvec{X} =
    \begin{bmatrix}
      0 & 1 \\
      1 & 0
    \end{bmatrix},\
  \mlvec{Y} =
    \begin{bmatrix}
            0 & -i \\
            i & 0 \\
          \end{bmatrix},\
  \mlvec{Z} =
    \begin{bmatrix}
      1 & 0 \\
      0 & -1 \\
    \end{bmatrix}.
    $
    For Pauli matrix $\mlvec{Z}$,
    $\exp({-i\theta\mlvec{Z})}=\begin{bmatrix}
      e^{-i\theta} & 0 \\
      0 & e^{i\theta}
    \end{bmatrix}$. 
    }
We can also group gates in QNNs with respect to the layer structure in Figure~\ref{fig:neural-network} by putting gates that can be executed in parallel in the same layer. For example, let $\mlvec V_{i,j} (\theta_{i,j})= \exp({-i\theta_{i,j} \mlvec H_{i,j}})$ be the $j$th gate in the $i$th layer. Then $\mlvec U_i (\mlvec \theta)=\prod_j \mlvec V_{i,j}(\theta_{i,j})$ and 
\begin{equation}
  \mlvec{U}(\mlvec{\theta}) = \mlvec{U}_t(\mlvec{\theta}) \mlvec{U}_{t-1}(\mlvec{\theta})\cdots
\mlvec{U}_1(\mlvec{\theta}),  
\end{equation}
where $\mlvec U(\mlvec \theta)$ refers to the unitary transformation of the entire QNN with parameters $\mlvec \theta$. 
For technical convenience and to highlight the dependence on the number of parameters $p$, we can also write
\begin{align}
  \label{ln:ansatz}
\mlvec{U}(\mlvec{\theta})=
\mlvec{V}_p({\theta_p})
\mlvec{V}_{p-1}({\theta_{p-1}})
\cdots
\mlvec{V}_1({\theta_1}),
\end{align}
 with $\mlvec V_l(\theta_l)= e^{-i\theta_l \mlvec H_l}$ for Hamiltonian $\mlvec{H}_l$ and $l\in[p]$.

%%%%%%%%%%%%%%%%%%%%%%%%%%%%%%%%%%%%%%%%%%%%%%%%%%%%%%%%%%%%
%%%%%%%%%%%%%%%%%%%%%%%%%%%%%%%%%%%%%%%%%%%%%%%%%%%%%%%%%%%%

\textbf{Quantum measurements and observables.} Quantum \emph{observables}, mathematically formulated as Hermitian matrices $\mlvec{M}\in\complex^{d\times
  d}$, are used in quantum mechanics to encode the information of the classical random outcomes generated by quantum \emph{measurements} on quantum states. 
The expected outcome $\hat{y}$ of observable $M$ on the output state $\mlvec{U}(\mlvec{\theta}) \mlvec\rho \mlvec{U}^\dagger(\mlvec{\theta})$ of any QNN $\mlvec U(\theta)$ is given by 
\begin{align*}
  \hat{y}= &f(\mlvec{\rho},\mlvec\theta) = \tr( \mlvec{U}(\mlvec{\theta}) \mlvec\rho \mlvec{U}^\dagger(\mlvec{\theta})\mlvec{M})\\
  \text{ or } \quad &\tr( \mlvec v^\dagger \mlvec{U}^\dagger(\mlvec{\theta})\mlvec{M}
\mlvec{U}(\mlvec{\theta}) \mlvec v) \text{ when } \mlvec \rho = \mlvec v\mlvec v^\dagger.
\end{align*}

A more complete introduction to quantum mechanics and QNNs can be found in \supref{sec:app_prelim}.

Given a quantum training set $\S=\{(\mlvec{\rho}_i, y_i)\}^m_{i=1}$ and a QNN with output $\hat{y}=f(\mlvec{\rho},\mlvec\theta)$, the empirical risk minimization with square loss optimizes the following loss function: 

\begin{align}
  \label{eqn:loss}
  L(\mlvec\theta; \S) = \frac{1}{m} \sum_{i=1}^m \left( \tr(\mlvec U(\mlvec\theta)\mlvec \rho_i\mlvec U^\dagger(\mlvec \theta)\mlvec M) - y_i \right)^2.
\end{align}
When quantum encoding states are pure, namely  $\mlvec{\rho}_i=\mlvec{v}_i\mlvec{v}_i^\dagger$ for $i\in [m]$, the loss function becomes 
\begin{equation}
 L(\mlvec\theta; \S) = \frac{1}{m} \sum_{j=1}^m \left( \mlvec{v}_j^\dagger\mlvec U^\dagger(\mlvec \theta)\mlvec
  M \mlvec U(\mlvec\theta)\mlvec{v}_j - y_j \right)^2
\end{equation}
which resembles Eqn.(\ref{eqn:class}) from one-hidden layer neural networks with quadratic activation except for unitary transformations.  
It is known in \cite{du2018power} that such neural
networks do not possess spurious local minima almost certainly, whereas we establish a completely different behavior for QNNs. 

%%%%%%%%%%%%%%%%%%%%%%%%%%%%%%%%%%%%%%%%%%%%%%%%%%%%%%%%%%%%%%%%%%%%%%%%%
%%%%%%%%%%%%%%%%%%%%%%%%%%%%%%%%%%%%%%%%%%%%%%%%%%%%%%%%%%%%%%%%%%%%%%%%%

\paragraph{Characterization of the landscape.}
For a differentiable function $F$ defined on an unconstrained domain,
$\mlvec \theta^*$ is a \emph{critical} point if and only if
the gradient vanishes at the point: $\nabla F(\mlvec\theta^*) = \mlvec 0$.
$\mlvec\theta$ is a local minimum if and only if
there is an open set $U$ containing $\mlvec \theta^*$ such that $F(\mlvec\theta^*)
\leq F(\mlvec\theta)$ for all $\mlvec\theta\in U$. A local minimum is global if
the minimum value of $F$ is attained at $\theta^*$. For twice-differentiable
function over an unconstrained domain, $\mlvec\theta^*$ is a local minimum if the
Hessian is positive definite at $\mlvec\theta^*$ (sufficient condition) and only
if $\mlvec\theta^*$ is a critical point (necessary condition).

Note further that the form of quantum gates $\exp({-i\theta \mlvec H})$ will be \emph{periodic} in $\theta$ for $\mlvec H$ with rational eigenvalues, which is typically true for commonly used $\mlvec H$ (e.g., Pauli matrices).  
It hence suffices to study the number of (spurious) local minima of the loss function within one period.

\section{Exponentially Many Spurious Local Minima for Under-parameterized QNNs}
\label{sec:construction}
In this section, we present our main result on the constructions of datasets for
$p$-parameter quantum neural network instances with $\Omega(2^p)$ spurious local
minima. We consider QNNs defined in Eqn.~(\ref{ln:ansatz}), with parameterized gates $\mlvec{V}_{l}(\theta_l)$
generated by $\mlvec{H}_l$ with eigenvalues $\pm 1$.
This is the case for single-qubit parameterized gates and
two-qubit gates generated by Kronecker products of Pauli matrices.

Shifting $\mlvec{H}_l$ by $\lambda\mlvec{I}$ for any $\lambda\in\real$
introduces a global phase factor to the output state and does not change the
output $f(\mlvec{\rho},\mlvec{\theta})$. Also, shifting the observable
$\mlvec{M}$ by $\lambda\mlvec{I}$ is equivalent to shifting the labels in the dataset by
$-\lambda$. Without loss of generality, we assume $\tr(\mlvec{H}_l) = 0$
and $\tr(\mlvec{M}) = 0$.

We start by characterizing the output $f(\mlvec{\rho},\mlvec{\theta})$.
For any $l\in[p]$, define linear maps $\Phi^{(0)}_l(\cdot)$,
$\Phi^{(1)}_l(\cdot)$ and $\Phi^{(2)}_l(\cdot)$ such that
\begin{align}
  \Phi^{(0)}_l(\mlvec{A}) &= \frac{1}{2}(\mlvec{A} + \mlvec{H}_l\mlvec{A}\mlvec{H}_l)\\
  \Phi^{(1)}_l(\mlvec{A}) &= \frac{1}{2}(\mlvec{A} - \mlvec{H}_l\mlvec{A}\mlvec{H}_l)\\
  \Phi^{(2)}_l(\mlvec{A}) &= \frac{i}{2}[\mlvec{H}_l, \mlvec{A}]
\end{align}
Here $[\cdot,\cdot]$ is the commutator of two matrices. For any Hermitian
$\mlvec{A}$, $\Phi^{(0)}_l(\mlvec{A})$ commutes with $\mlvec{H}_l$, and the
output of $\Phi^{(1)}_l$ and $\Phi^{(2)}_l$ anti-commute with $\mlvec{H}_l$.
For any vector $\mlvec\xi \in \{0,1,2\}^p$, define:
\begin{align}
  \Phi_{\mlvec{\xi}}(\mlvec{A})= \Phi_{1}^{(\xi_1)}\circ\Phi_{2}^{(\xi_2)}\circ\cdots\circ\Phi_{p}^{(\xi_p)}(\mlvec{A})
\end{align}
with $\circ$ denoting the composition of mappings.

The observable in Heisenberg picture $\mlvec{M}(\mlvec{\theta}) :=
\mlvec{U}^\dagger(\mlvec{\theta})\mlvec{M}\mlvec{U}(\mlvec{\theta})$ can be
expanded as:
\begin{align}
  \sum_{\mlvec{\xi}\in\{0,1,2\}^p}\Phi_{\mlvec{\xi}}(\mlvec{M})
  \prod_{l:\xi_{l}=1}\cos2\theta_{l}
  \prod_{l^\prime:\xi_{l^\prime}=2}\sin2\theta_{l^\prime}
\end{align}
The QNN output $f(\mlvec{\rho}, \mlvec{\theta}) =
\tr(\mlvec{\rho}\mlvec{M}(\mlvec{\theta}))$ can be expressed as the following trigonometric polynomial:  
\begin{align}
  \sum_{\mlvec{\xi}\in\{0,1,2\}^p}\tr(\mlvec{\rho}\Phi_{\mlvec{\xi}}(\mlvec{M}))
  \prod_{l:\xi_{l}=1}\cos2\theta_{l}
  \prod_{l^\prime:\xi_{l^\prime}=2}\sin2\theta_{l^\prime}
  \label{eqn:1sample}
\end{align}
As shown in \supref{sec:app_constructions}, the loss function remains
invariant under the joint transformation $\theta_l\mapsto \theta_l + \frac{\pi}{2}$ and
\begin{align}
  \Phi^{(0)}_{l}(\cdot)\mapsto&\mlvec{H}_l\Phi^{(0)}_{l}(\cdot)\mlvec{H}_l =\Phi^{(0)}_{l}(\cdot)\\
  \Phi^{(1)}_{l}(\cdot)\mapsto&\mlvec{H}_l\Phi^{(1)}_{l}(\cdot)\mlvec{H}_l =-\Phi^{(1)}_{l}(\cdot)\\
  \Phi^{(2)}_{l}(\cdot)\mapsto&\mlvec{H}_l\Phi^{(2)}_{l}(\cdot)\mlvec{H}_l =-\Phi^{(2)}_{l}(\cdot)
\end{align}
Under the transformation $\theta_l\mapsto\theta_l+\frac{\pi}{2}$, terms in Eqn.~(\ref{eqn:1sample}) associated with $\mlvec{\xi}:\xi_l=0$ are invariant, while terms 
associated with $\mlvec{\xi}:\xi_l=1,2$ flip signs.

From an alternative perspective,
$L(\mlvec\theta; \S)$ contains \emph{oscillating wave} components 
proportional to $\cos 4\theta_l$, $\sin 4\theta_l$, $\cos2\theta_l$ and
$\sin2\theta_l$, hence periodic with $\pi$ on each coordinate. However, due the existence of
lower frequency, the periodicity with $\frac{\pi}{2}$ does not always hold for
all datasets. Our construction utilizes the presence and absence of this
$\frac{\pi}{2}$-translational symmetry.

We will focus on a general class of QNN, which we call QNN \emph{with
linear independence}:
\begin{definition}[QNN with linear independence]
  \label{def:linear_indep}
  A QNN is said to be with linear independence, if the
  associated set of $3^p-1$ operators $\{\Phi_{\mlvec{\xi}}(\mlvec{M})\}_{\mlvec\xi\in\{0,1,2\}^p,
  \mlvec\xi\neq \mlvec{0}}$ forms a linearly independent set.
\end{definition}
Note that for the linear independence condition to hold, the dimension of the
QNN $d \geq 3^{p/2}$. Namely, it is a
under-parameterized case, which differentiates us from the over-parameterized
ones~\cite{rabitz2004quantum,russell2016quantum,kiani2020learning}.
Our main result states:
\begin{restatable}[Construction: exponentially many local minima]{theorem}{constructthm}
  \label{thm:construction}
  Consider QNNs composed of unitaries generated by two-level
  Hamiltonians, parameterized by $\mlvec\theta \in \real^p$.
  If the QNN is with linear independence, a dataset $\S$
  can be constructed to induce a loss function $L(\mlvec{\theta};\S)$ 
  with $2^p$ local minima within each period, and $2^p-1$ of these minima
  are spurious with positive suboptimality gap. 
\end{restatable}
\begin{proof}[Proof of Theorem~\ref{thm:construction}]
    The dataset we construct is composed of two parts $\S_0$ and $\S_1$. The
    first component of the loss function $L(\mlvec\theta; \S_0)$ is constructed
    with $2^p$ local minima using the $\frac{\pi}{2}$-translational symmetry:
    \begin{restatable}[Creating symmetry]{lemma}{constructsymlm}
      \label{lm:construction_sym}
      For QNNs with linear independence as mentioned in
      Theorem~\ref{thm:construction}, a dataset $\S_0$ can be constructed to
      induce a loss function $L(\mlvec{\theta};\S_0)$ that (1) has a local minimum at
      some $\mlvec{\theta}^\star$, and (2) is invariant under translation
      $\theta_l\mapsto\theta_l+\frac{\pi}{2}$ for all $l\in[p]$. 
    \end{restatable}
    Due to the translational invariance,
    for any $\mlvec{\zeta}\in\{0,1\}^p$, $\mlvec{\theta}^\star
    + \frac{\pi}{2}\mlvec{\zeta}$ is a local minimum for $L(\mlvec\theta;\S_0)$,
    forming a total of $2^p$ local minima. A second dataset $\S_1$ is introduced to break this  symmetry,
    creating spurious local minima:
\begin{restatable}[Breaking symmetry]{lemma}{constructbreaksymlm}
  \label{lm:construction_breaksym}
  Consider the QNN, dataset $\S_0$ and local minimum $\mlvec{\theta}^\star$ defined in
  Lemma~\ref{lm:construction_sym}. Let $\Theta$ denote the set of $2^p$ local
  minima due to the translational invariance.
  There exists a dataset $\S_1$ such that
  \begin{multline}
    \inf_{\mlvec\theta\in\N(\mlvec{\theta}^\star)}
    L(\mlvec{\theta}; \S_0) + L(\mlvec{\theta}; \S_1) <\\
    \inf_{\mlvec\theta\in\N(\mlvec{\theta}^\prime)}
    L(\mlvec{\theta}; \S_0) + L(\mlvec{\theta}; \S_1)
    \label{eqn:break_sym2}
  \end{multline}
  for all $\mlvec{\theta}^\prime\in\Theta/\{\mlvec{\theta}^\star\}$, and that
  \begin{align}
    L(\mlvec{\theta}; \S_0) + L(\mlvec{\theta}; \S_1) >
    L(\mlvec{\theta}^\prime; \S_0) + L(\mlvec{\theta}^\prime; \S_1)
    \label{eqn:break_sym1}
  \end{align}
  for all $\mlvec{\theta}^\prime\in\Theta$ and all 
  $\mlvec\theta\in\partial\N(\mlvec{\theta}^\prime)$.
  Here $\N(\cdot)$ denote a bounded and closed neighbourhood, such that $\N(\mlvec\theta)\cap\N(\mlvec\theta^\prime)=\emptyset$ for any
  $\mlvec\theta,\mlvec\theta^\prime \in\Theta$. And let $\partial\N$ denote its
  boundary.
  \end{restatable}
Eqn.~(\ref{eqn:break_sym1}) in Lemma~\ref{lm:construction_breaksym} ensures the
existence of a local minimum within $\N(\mlvec{\mlvec{\theta}})$ for each
$\mlvec{\theta}\in\Theta$, and Eqn.~(\ref{eqn:break_sym2}) promises that only
the local minimum within $\N(\mlvec{\theta}^\star)$ achieves the global optimal
value. Combining $\S_0$ and $\S_1$ finishes the proof for Theorem~\ref{thm:construction}.
\end{proof}
We give proof sketches for Lemma~\ref{lm:construction_sym} and
\ref{lm:construction_breaksym}. The full proofs are postponed to \supref{sec:app_constructions}.
\begin{proof}[Proof sketch for Lemma~\ref{lm:construction_sym}]
  It suffices to construct a dataset $\S_0 =
  \{(\mlvec\rho_k,y_k)\}_{k=1}^{m_0}$, such that (1) for all $k\in[p]$, $f_k(\mlvec{\theta}):=\<\mlvec\rho_k,
  \mlvec{M}(\mlvec\theta)\> - y_k$ is either symmetric or anti-symmetric under
  $\theta_l\mapsto\theta_l + \frac{\pi}{2}$ for all $l\in[p]$, and (2) 
  the intersection $\Theta$ of the set of roots $\Theta_k$ of $f_k(\mlvec{\theta})
  = 0$ is non-empty and contains at least one isolated point
  $\mlvec{\theta}^\star$. For such $\S_0$, $\mlvec{\theta}^\star$ is an isolated
  root of the non-negative loss function $L(\mlvec{\theta}; \S_0) =
  \sum_{k=1}^{m_0}f_k(\mlvec{\theta})^2$.

  The existence of such dataset $\S_0$ follows from the linear independence of
  operators for the QNN. As a result, 
  for any $k\in[m_0]$, the solution to the following linear system for Hermitian $\mlvec{D}_k\in
  \complex^{d\times d}$ is non-empty:
  \begin{align}
    \begin{cases}
      &\tr(\mlvec{D}_k\cdot\mlvec{I}) = 0, \\ 
      &\tr(\mlvec{D}_k\cdot\Phi_{\mlvec{\xi}}(\mlvec{M})) = \hat{f}_{\mlvec{\xi},k},\ \forall \mlvec{\xi}\neq\mlvec{0}. 
    \end{cases}
  \end{align}
  Here $\hat{f}_{\mlvec{\xi},k}$ is the coefficient corresponding to the
  term $\prod_{l:\xi_l=1}\cos2\theta_l\prod_{l^\prime:\xi_{l^\prime}=2}\sin2\theta_{l^\prime}$ in $f_k(\mlvec\theta)$.
  Given the solution $\{\mlvec{D}_k\}_{k=1}^{m_0}$, $\S_0$ can be constructed by setting
  $\mlvec{\rho}_k:=\frac{1}{d}\mlvec{I} + \kappa\mlvec{D}_k$ for a proper
  scaling factor $\kappa$ and let $\ y_k =
  \tr(\mlvec{\rho}_k\Phi_{\mlvec{0}}(\mlvec{M}))$.
\end{proof}
\begin{proof}[Proof sketch for Lemma~\ref{lm:construction_breaksym}]
  Rewrite the loss function as
  \begin{align}
    L(\mlvec\theta; \S_1) =&-\frac{2}{m_1}\sum_{k=1}^{m_1}y_k\tr(\mlvec\rho_k\mlvec{M}(\mlvec\theta))\\
                           &+ \frac{1}{m_1}\sum_{k=1}^{m_1}\tr(\mlvec{\rho}_k \mlvec{M}(\mlvec\theta))^2
                             + \frac{1}{m_1}\sum_{k=1}^{m_1}y_k^2
  \end{align}
  As will be made clear in \supref{sec:app_constructions}, our key observation
  is that, under a joint scaling of $y_k$ and $\mlvec{\rho}_k$, the second
  term can be arbitrarily suppressed while the first term remains the same. Therefore it suffices to
  study the first term ${L}^\prime(\mlvec{\theta}; \S_1) :=
  -\frac{2}{m_1}\sum_{k=1}^{m_1}y_k\tr(\mlvec\rho_k\mlvec{M}(\mlvec\theta))$.
  The linear independence allows us to solve a linear system to construct $\S_1$
  that satisfies the
  requirements in Lemma~\ref{lm:construction_breaksym}.
\end{proof}
\textbf{Remarks.} The statements above involve unitaries generated by two-level
Hamiltonians only.  For more
general local quantum gates,  $\{\mlvec{H}_l\}_{l=1}^p$ are allowed to have more than two
distinct eigenvalues. We are especially interested in Hamiltonians with eigenvalues
$\{E_1,\cdots, E_d\}\subset\mathbb{Z}$, as arbitrary Hamiltonians with rational
spectrum can be converted to ones with integral spectrum with proper shifting and scaling.
Theorem~\ref{thm:construction} can be generalized for $\mlvec{H}_l$'s with
largest eigen-gap $\max_{c,c^\prime\in[d]}|E_c - E_{c^\prime}|$ bounded by
$\Delta$, with the number of spurious local minima being $\Omega(\Delta^p)$.
This observation further supports the intuition of interference as the source of local minima. 
\paragraph{1-layer QNN.}
 A simple example of QNNs with linear
independence is a one-layer circuit with local $\mlvec{H}_l$ acting on the $l$-th qubit,
and a product operator $\mlvec{M}$ as the observable:
\begin{restatable}[One-layer QNNs with product observables]{proposition}{onelayerprop}
  \label{prop:onelayer}
  Consider the family of QNNs composed of unitaries generated by two-level
  Hamiltonians, parameterized by $\mlvec\theta \in \real^p$.
  For all $l\in[p]$, let $\mlvec{H}_l$ be a local Hamiltonian on the $l$-qubit,
  taking the form
  $\mlvec{I}\otimes \cdots\otimes\mlvec{h}_l\otimes\cdots\otimes\mlvec{I}$ for
  some Hermitian $\mlvec{h}_l$ at the $l$-th position, and $\mlvec{M} =
  \mlvec{m}_1\otimes\cdots\otimes\mlvec{m}_p$ such that $\mlvec{m}_l +
  \mlvec{h}_l\mlvec{m}_l\mlvec{h}_l$ and $\mlvec{m}_l -  
  \mlvec{h}_l\mlvec{m}_l\mlvec{h}_l$  are
  non-zero for any $l$.
  There exists a dataset that induces a loss function with $2^p-1$ spurious
  local minima.
\end{restatable}
This follows from the fact that $\tr({\Phi}_{\mlvec{\xi}}(\mlvec{M})
{\Phi}_{\mlvec\xi^\prime}(\mlvec{M})) = 0$ if and only if $\mlvec{\xi} \neq \mlvec{\xi^\prime}$.
In \supref{sec:app_constructions},
we provide proof for Proposition~\ref{prop:onelayer} and several concrete
example QNNs to demonstrate that our construction can have local minima at arbitrary $\mlvec{\theta}$, and does not allow
trivial solutions such as coordinate-wise greedy optimization.

\section{Typical QNNs are with Linear Independence}
\label{sec:linear_indep}
In Section~\ref{sec:construction}, we provided a general condition
(Definition~\ref{def:linear_indep})
for QNNs to have exponentially many bad
local minima for some datasets.
In this section, we show that this condition is met for typical
under-parameterized QNNs. To see this, we consider the following
measure over instances of QNNs:
Let $\mlvec{H}$ be a $d$-dimensional Hermitian such that $\tr(\mlvec{H}) = 0$
and $\mlvec{H}^2 = \mlvec I$. A random circuit $\mlvec{U}(\mlvec{\theta})$ is
specified as
\begin{align}
  \mlvec{U}(\mlvec{\theta}) =
  e^{-i\theta_p\mlvec{W}_p\mlvec{H}\mlvec{W}_p^\dagger}
  \cdots
  e^{-i\theta_1\mlvec{W}_1\mlvec{H}\mlvec{W}_1^\dagger}
  \label{eqn:random_model}
\end{align}
with $\{\mlvec{W}_l\}_{l=1}^p$ independently sampled with respect to the Haar measure on the $d$-dimensional unitary
group $U(d)$.

Up to a unitary transformation, this random model is
equivalent to a circuit with $p$ interleaving parameterized gate
$\{e^{-i\theta_l\mlvec{H}}\}_{l=1}^p$ and unitary $\{\tilde{\mlvec W}_l\}_{l=1}^p$
randomly sampled with respect to the Haar measure:
\begin{align}
  \mlvec{U}(\mlvec{\theta}) = \tilde{\mlvec{W}}_pe^{-i\theta_p\mlvec{H}}\tilde{\mlvec{W}}_{p-1}\cdots\tilde{\mlvec{W}}_1e^{-i\theta_1\mlvec{H}}
  \label{eqn:random_model1}
\end{align}
The equivalence is due to the left (or right) invariance of the Haar measure. This interleaving nature of fixed and parameterized gates are shared by existing designs of QNNs, and any $p$-parameter QNN generated by two-level Hamiltonians can be expressed in Eqn.~(\ref{eqn:random_model1}). Moreover, applying polynomially many random 2-qubit gates on random pairs of qubits generates a distribution over gates that approximates the Haar measure up to the $4$-th moments \cite{brandao2016local}, which is what we require in this section.

The Gram matrix for the set $\{\Phi_{\mlvec\xi}(\mlvec M)\}_{\mlvec\xi\in\{0,1,2\}^p,
  \mlvec\xi\neq\mlvec{0}}$ is defined such that the element corresponding to
the pair $(\mlvec{\xi},\mlvec{\xi}^\prime)$ is
$\tr(\Phi_{\mlvec{\xi}}(\mlvec{M})\Phi_{\mlvec{\xi}^\prime}(\mlvec{M}))$. The
Gram matrix is always positive semidefinite, and a positive definite Gram matrix
implies the linear independence of the set.

Using the integral formula with respect to Haar measure on unitary groups
\cite{puchala2011symbolic}, we can
estimate the expectations and variances of the diagonal and off-diagonal
terms, and upper bound the probability of the event:
\begin{align}
  \exists \mlvec\xi: \tr(\Phi_{\mlvec{\xi}}(\mlvec{M})^2)
  \leq \sum_{\mlvec{\xi}^\prime\neq \mlvec{\xi}} |\tr(\Phi_{\mlvec{\xi}}(\mlvec{M})
  \Phi_{\mlvec{\xi}^\prime}(\mlvec{M}))|
\end{align}

Applying the Gershgorin circle theorem \cite{golub1996matrix}, we can lower bound the
probability for a random QNN to have linear independent terms:
\begin{restatable}[Typical under-parameterized QNNs are with linear independence]{theorem}{randomthm}
  \label{thm:random_qnn}
  Consider a random $p$-parameter $d$-dimensional QNN with two-level
  Hamiltonians sampled from the model specified in Eqn.~(\ref{eqn:random_model}).
  Let the observable $\mlvec{M}$ be an arbitrary non-zero trace-$0$ Hermitian.
  Such QNN is with linear independence with probability $\geq 1 - O(d^{-1})$ for fixed $p$, and with
  probability $\geq 1 - O(e^{-p})$ for dimension $d: \log(d) = \Theta(p)$.
\end{restatable}

Please refer to \supref{sec:app_linear_indep} for the full proof.

\section{Upper Bound on the Number of Local Minima}
\label{sec:upperbound}
Our construction above possesses $2^p$ local minima for $p$ parameters, whereas
the classical work of \citet{auer1996exponentially} demonstrates a construction
for a single neuron with $\lfloor m/p\rfloor^p$ local minima for $m$ training samples.
Note that the latter could grow unboundedly with $m$.
In this section, we show, however, this classical unbounded growth of local minima does not hold for QNNs. 
In fact, we could establish an almost matching upper bound for $2^p$. All the formal proofs are deferred to \supref{sec:app_upperbound}.

To that end, let us examine the \emph{Fourier} expansion of the loss function
$L(\mlvec{\theta}, \S)$ (Eqn.~(\ref{eqn:loss})).
Let $T_l$ be the period of $L(\mlvec{\theta}; \S)$ corresponding to $\theta_l$, and 
$\hat{L}(\mlvec k)$ the Fourier coefficient for  $\mlvec k = (k_1, \cdots, k_p)^T\in \mathbb{Z}^p$. We have
\begin{align}
  L(\mlvec\theta;\S)
  = \sum_{\mlvec k\in K}\hat{L}(\mlvec k)\prod_{l=1}^p\Bigl(\cos\frac{k_l\theta_l}{T_l}+ i\sin\frac{k_l\theta_l}{T_l}\Bigr)
\end{align}
where $K \subseteq \mathbb{Z}^p$ is the support of the Fourier coefficients.

One critical observation is that,
for arbitrary choice of two-level $\{\mlvec{H}_l\}_{l=1}^p$, observable $\mlvec{M}$ and
training set $\S$, the support $K$ of the Fourier spectrum is bounded in
$\ell_1$-norm: $\max_{\mlvec k\in K}\sum_{l=1}^p|k_l| \leq 2p$, indicating that
the Fourier degree of $L(\mlvec{\theta};\S)$ is upper
bounded by $2p$ (See \supref{subsec:app_ft}). 

By definition, a local minimum must be a critical point, hence it suffices to bound
the number of critical points for functions with Fourier spectrum supported on a $\ell_1$-bounded set.
Define $G_l(\mlvec\theta)$ as $\frac{\partial}{\partial\theta_l}L(\mlvec\theta;\S)$:
\begin{align}
  G_l(\mlvec\theta)
  = &\sum_{\mlvec k\in K}k_l\hat{L}(\mlvec k)\bigl(-\sin\frac{k_l\theta_1}{T_l}+ i\cos\frac{k_l\theta_l}{T_l}\bigr)\\ 
                     &\quad\cdot\prod_{l^\prime\neq l} \bigl(\cos\frac{k_{l^\prime}\theta_{l^\prime}}{T_{l^\prime}}+ i\sin\frac{k_{l^\prime}\theta_{l^\prime}}{T_{l^\prime}}\bigr)
\end{align}
Notice that the Fourier spectrum of $G_l$ is supported on the same set $K$.
A critical point of $L(\mlvec{\theta};\S)$ must satisfy that for all
$l\in[p]$, $G_l(\mlvec\theta) = 0$.
By basic trigonometry, $\cos k\theta$ can be expressed as a degree-$k$ polynomial of
$\cos\theta$ and $\sin k\theta$ as a degree-$(k-1)$ polynomial of $\cos\theta$ multiplied by $\sin\theta$. Consider the change of variable
\begin{align}
  c_l = \cos (\theta_l/T_l),\ s_l = \sin({\theta_l}/{T_l}),\
  \forall l\in[p].
\end{align}
Let $g_l(c_1, s_1, \cdots, c_p, s_p)$ be the multivariate polynomial constraints
corresponding to $G_l(\mlvec\theta)$ after the change of variable. For each
$g_l$, the sum of degrees of $c_{l^\prime}$ and $s_{l^\prime}$ is bounded by
$\max_{\mlvec{k}\in K}|k_{l^\prime}|$, and the degree $\mathsf{deg}(g_l)$ of
$g_l$ is bounded by $\max_{\mlvec{k}\in}\sum_{l=1}^p|k_l|\leq 2p$.
The change of variable is one-to-one from $\theta_l\in[0,T_l)$ to a pair of $(c_l,
s_l)\in\real^2$ under the constraint $c_l^2 + s_l^2 = 1$. Therefore, it suffices
to count the number of roots of the polynomial system with $2p$ parameters
and $2p$ constraints: 
\begin{align}
  g_l(c_1, s_1, \cdots, c_p, s_p) = 0,\ c_l^2 + s_l^2 - 1 = 0
\end{align}
for all $l\in[p]$.
Notice that for a general polynomial system, the number of critical points can be
unbounded. For example, consider a system composed of constant polynomials,
every point in the domain is a critical point. This corresponds to the constant
loss function, where the gradients vanish everywhere with positive
semidefinite Hessians. For this reason, we will focus on the non-degenerated case
with finitely many local minima. Under the premise of non-degeneracy, 
by B\'ezout's Theorem (e.g. Section 3.3 in \citet{cox2006using}), the number of
roots can be bounded by
the product of the degree of polynomial constraints
$2^p\mathsf{deg}(g_1)\mathsf{deg}(g_2)\cdots\mathsf{deg}(g_p) \leq (4p)^p$.
A formal statement of the above derivation is as follows:
\begin{restatable}[Upper bound: the number of local minima]{theorem}{upperthm}
  \label{thm:upper}
  Consider non-degenerated QNNs composed of unitaries generated by two-level Hamiltonians $\{H_l\}_{l=1}^p$ with $p$ parameters. For 
  training set $\S$, within each period, the loss function $L(\mlvec\theta; \S)$ possesses at most
  $(4p)^p$ local minima.
\end{restatable}
We also prove a similar result for the more general case where the generators are
Hamiltonians with integral spectrum: let $\Delta$ be the largest eigen-gap for
each of the Hamiltonians, the number of local minima within each period is
upper bounded by $O((\Delta p)^p)$. Please refer to \supref{sec:app_upperbound} for details. 

\section{Experiments}
\label{sec:experiments}
We investigate the practical performance of the common optimizers on our construction in this section. 
It is well-known in the classical literature that the existence of spurious
local minima does not necessarily cause difficulties in optimization (e.g.,~\citet{ge2017optimization}). 
We show, however, our constructions can indeed be hard instances for training in practice.

To that end, we evaluate a specific construction from
Proposition~\ref{prop:onelayer} in Section~\ref{sec:construction} by using the
standard optimizers with randomly initialized parameters uniformly sampled from
the domain\footnote{For $p$-parameter instances, we uniformly sample the initial parameters from
$[0,2\pi)^p$.}, and visualize the distribution of function values at
convergence.
For $p$-parameter instances, our construction involves $p$-qubits. We choose
$\mlvec{h}_1 = \cdots = \mlvec{h}_p = \mlvec{Z}$ and
$\mlvec{m}_1=\cdots=\mlvec{m}_p = \mlvec{Y} + \mlvec{I}$. The specific form of the instance and all the training details are provided in \supref{sec:app_num}.

\paragraph{Implementation}The experiments are run on Intel Core i7-7700HQ
Processor (2.80GHz) with 16G memory. 
We classically simulate the training with
Pytorch \cite{paszke2019pytorch}, using the analytical form of the objective
function for the purpose of efficiency.

\paragraph{Optimizers}
The QNN instances are trained with three popular optimizers
in classical optimization or machine learning: \textsf{Adam}\cite{kingma2014adam},
\textsf{RMSProp}\cite{bengio2015rmsprop}, and \textsf{L-BFGS}\cite{liu1989limited}. The first two
methods~\cite{kingma2014adam,bengio2015rmsprop} are variants of vanilla gradient descent with adaptive learning rate.
and are widely used for training large-scale deep neural networks as well as for
the quantum counterparts \cite{killoran2019continuous,
  mari2019transfer,lloyd2020quantum, ostaszewski2019qcsl, sweke2019stochastic}.
The last method~\cite{liu1989limited} is an efficient implementation of the approximate Newton method
that utilizes the second-order information (i.e. the Hessian). For all instances
and optimizers, we use the exact gradient induced by the dataset without
stochasticity from the mini-batched gradient descent.

It turns out, for all the examined instances and all three optimizers, under random
initialization, the optimizations converge to local minima with
non-negligible suboptimality (i.e., different from the global one by a non-negligible amount) with high probability. In Figure~\ref{fig:dist0}, we train the $4$-parameter construction
with \textsf{RMSProp} and repeat for 100 times. Let $\mlvec\theta_i$ and $\mlvec\theta_f$
denote the parameters at initialization and at convergence. The function values at initialization
$L(\mlvec{\theta}_i;\S)$ are supported on a continuous spectrum as shown in gray. After training and converging with
\textsf{RMSProp}, 
the function values $L(\mlvec{\theta}_f;\S)$ fall into discretized
values as shown in orange.
The smallest training loss attainable in our construction is $0$,
therefore only the leftmost bar (to the left of the dotted \textbf{black} vertical
line) corresponds to the
global minimum.
Namely, the success probability of converging to the global minimum is very small. 
A similar phenomenon persists for instances with more parameters and with different optimizers in Figure~\ref{fig:dist}.
As the number of bad local minima grows exponentially in our construction, the
success probability should also in theory decay exponentially. This is
empirically confirmed in Figure~\ref{fig:decay}, where we illustrate the precise
empirical success probability for all these tests. Moreover, as shown in
\supref{subsec:app_noise}, the tendency of exponential decay remains unchanged
in the presence of label noises, indicating the robustness of our constructions. 

\begin{figure}[!htbp]
  \centering
    \includegraphics[width=.8\linewidth]{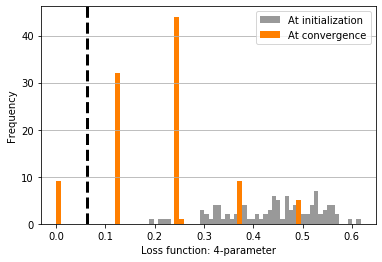}
  \caption{
    Loss functions at random initialization and at
    convergence for $4$-parameter instances trained with \textsf{RMSProp}, repeated
    for 100 times. The function values are supported on a continuous spectrum at
    initialization as plotted in \textbf{\color{gray} gray} and converge to discretized values as
    plotted in \textbf{\color{orange} orange}. }
  \label{fig:dist0}
\end{figure}
\begin{figure}[!htbp]
  \centering
    \subfigure{
    \includegraphics[width=.8\linewidth]{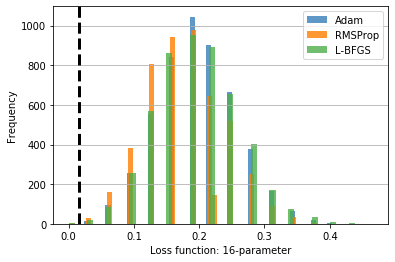}
  }
  \caption{
    Distributions of loss functions at convergence for instances with
    $16$ parameters trained with \textsf{Adam}, \textsf{RMSProp} and \textsf{L-BFGS},
    repeated 5000 times with uniformly random initialization. 
    All methods fail to converge to the global minimum 0.0 with high probability.}
  \label{fig:dist}
\end{figure}

\begin{figure}[!htbp]
  \centering
    \includegraphics[width=.8\linewidth]{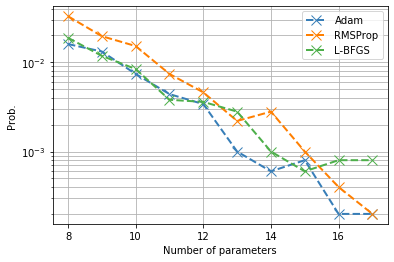}
    \caption{
      The decay of success rate for finding the global minimum under random
      initialization with \textsf{Adam}, \textsf{RMSProp}, \textsf{L-BFGS}. For each data
      point, we repeat the experiments for 5000 times. 
  }
  \label{fig:decay}
\end{figure}

\paragraph{Beyond the constructed datasets}
To demonstrate the generality of our results, we repeat the experiments
for datasets with more practical significance: for $p$-parameter
instances, we choose the input state to be a $p$-qubit encoding of
$\mlvec{x}\in[0,2\pi)^{2p}$ via $\mlvec{X}$- and $\mlvec{Y}$-rotations on each of the qubits.
The associated label is either $1$ or $0$, depending on the sign of
$\mlvec{w}^T\mlvec{x}$, with $\mlvec{w}$ being the normal vector of a hyperplane in $\real^{2p}$.
These datasets have the interpretation as an encoding of a linearly
separable classical concept. In Figure\ref{fig:8-qubit-common}, we plot the function
values at convergence for an $8$-parameter instance: no more than 4 of the 70 random initializations have reached the
global minima. This is repeated for instances with $2,4$ and $6$ qubits. While we no
longer have a clear exponential dependency in the success rate, the number of local
minima increases significantly as the number of parameters increases (see
\supref{subsec:app_common_datasets}). 
This observation suggests that our theory and experiments on the constructed 
datasets can capture the practical difficulty in training under-parameterized
QNNs with gradient-based methods.

\begin{figure}[!htbp]
  \centering
  \includegraphics[width=\linewidth]{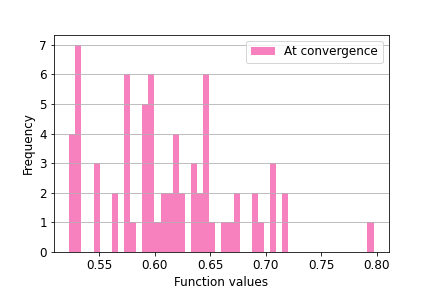}
  \caption{Function values at convergence for training an 8-parameter instance
    with \textsf{RMSProp} on the linearly-separable classical concept. No more than 4 among the 70 random
    initializations find the global minima, indicating the existence of many
    sub-optimal local minma.}
  \label{fig:8-qubit-common}
\end{figure}

\section{Conclusion}
\label{sec:conclusion}
In this work, we provide a characterization of the landscape for
under-parameterized QNNs, by showing that in the worst-case, the number of local
minima can increase exponentially with the number of parameters. Supported by
numerical simulations, our result suggests when under-parameterized, QNNs may not be efficiently solved by
gradient-based black-box methods.

This work leaves several open questions:
\begin{itemize}
\item Given the knowledge of the data distribution, can we design a QNN architecture with a benign landscape?
\item We know that when sufficiently parameterized (e.g.
  \cite{russell2016quantum}), the landscape for optimizing variational quantum ansatz can be benign. It is therefore natural
  to ask, fixing the system size, how does the landscape change as the number of
  parameters increases?
\item
  Classically, despite the provable bad landscape of shallow neural
  networks(e.g. \citet{safran2017spurious}), \citet{goel2019learning}
  came up with algorithms that can minimize the loss with guarantees. Can we design an algorithm (beyond gradient-based method) that can solve the
  optimization problem efficiently and provably?
 \end{itemize}

\section*{Acknowledgements}
We thank reviewers for useful comments.
This work received support from the U.S. Department of Energy,
Office of Science, Office of Advanced Scientific Computing
Research, Accelerated Research in Quantum
Computing and Quantum Algorithms Team programs, as well as the U.S. National Science Foundation grant CCF-1755800, CCF-1816695, and CCF-1942837
(CAREER).

\bibliographystyle{icml2021}

\clearpage

% Remove these appendix to a different file
\appendix
\renewcommand{\theequation}{\thesection.\arabic{equation}}
\setcounter{equation}{0}
\section{Preliminaries}
\label{sec:app_prelim}
In this section we introduce necessary backgrounds on quantum information and
linear algebra. For further readings, we refer the readers to the books by
Nielsen and Chuang~\cite{nielsen2002quantum} and by Watrous~\cite{watrous2018theory}.
\subsection{Linear Algebra}
\textbf{Basic notions.}  We use $\mlvec{I}_{d\times d}$ to denote the identity
matrix in $\complex^{d\times d}$. The subscripts are omitted when there is
no ambiguity on the dimension. Let $\dagger$ denote the conjugate transpose of
complex matrices.
A matrix $\mlvec{U}\in\complex^{d\times d}$ is said to be a unitary if
$\mlvec{U}\mlvec{U}^\dagger = \mlvec{I}_{d\times d}$.
A matrix is said to be Hermitian if its conjugate transpose is itself:
$\mlvec{A} = \mlvec{A}^\dagger$. Let $[\cdot, \cdot]$ denote the commutator of
two matrices, such that $[\mlvec{A},\mlvec{B}]:=\mlvec{A}\mlvec{B}-\mlvec{B}\mlvec{A}$.

\textbf{Inner product and adjoints.} The inner product between two Hermitians
$\mlvec{A}$ and $\mlvec{B}$ is defined as:
\begin{align}
  \<\mlvec A, \mlvec B\> := \tr(\mlvec{A}\mlvec{B})
\end{align}
For a linear map $\Phi(\cdot)$, the adjoint of the linear map $\Phi^*(\cdot)$
is defined such that
\begin{align}
  \<\mlvec{A}, \Phi(\mlvec{B})\> = \<\Phi^*(\mlvec{A}), \mlvec{B}\>
\end{align}
The Frobenius norm of a Hermitian $\mlvec{A}$ is defined as:
\begin{align}
  \|\mlvec{A}\|_F := \sqrt{\<\mlvec{A}, \mlvec{A}\>}
\end{align}
\textbf{Exponent of matrices.}
The exponent of a Hermitian $\mlvec{A}$ is defined using Taylor expansion
\begin{align}
\exp(\mlvec{A}):=\sum_{k=0}^{\infty}\frac{{\mlvec{A}}^{k}}{k!},
\end{align}
Let $\{E_j\}_{j=1}^d$ and $\{\mlvec{u}_j\}_{j=1}^d$ be the eigenvalues and
eigenvectors of $\mlvec{A}$,
We have
\begin{multline}
  \exp(\mlvec{A}):=\sum_{k=0}^{\infty}\frac{{\mlvec{A}}^{k}}{k!}\\
  = \sum_{k=0}^\infty\sum_{j=1}^d\frac{1}{k!} E_j^k \mlvec{u}_j\mlvec{u}_j^\dagger
  = \sum_{j=1}^d e^{E_j}\mlvec{u}_j\mlvec{u}_j^\dagger,
\end{multline}
For $\theta\in\real$ and Hermitian $\mlvec{H}$, $\exp(-i\theta\mlvec{H})$ is
unitary, since:
\begin{align}
  \exp(-i\theta\mlvec{H})^\dagger
  &= \Bigl( \sum_{k=0}^\infty \frac{\bigl( -i\theta \mlvec H \bigr)^k}{k!} \Bigr)^\dagger \\
  &= \sum_{k=0}^\infty \frac{\bigl( i\theta \mlvec H \bigr)^k}{k!}
  = \exp(i\theta\mlvec{H}),
\end{align}
and
\begin{align}
  &\exp(i\theta\mlvec{H})\exp(-i\theta\mlvec{H})
  =\bigl( \sum_{j=1}^d e^{i\theta E_j}\mlvec{u}_j\mlvec{u}_j^\dagger\bigr)\\
  =&\bigl( \sum_{k=1}^d e^{-i\theta E_k}\mlvec{u}_k\mlvec{u}_k^\dagger\bigr)
  =\mlvec{I}_{d\times d}.
\end{align}

\textbf{Kronecker product.}
For $\mlvec{A}\in\complex^{d_1\times d_2}$ and $\mlvec{B}\in\complex^{d_2\times
  d_2}$, the Kronecker product of $\mlvec{A}$ and $\mlvec{B}$ is defined as:
\begin{align}
  \mlvec{A}\otimes \mlvec{B} =
  \begin{bmatrix}
    A_{11}\mlvec{B}   & A_{12}\mlvec{B}   & \cdots & A_{1d_1}\mlvec{B}  \\
    A_{21}\mlvec{B}   & A_{22}\mlvec{B}   & \cdots & A_{2d_1}\mlvec{B}  \\
    \vdots            & \vdots           & \ddots & \vdots            \\
    A_{d_11}\mlvec{B} & A_{d_1 2}\mlvec{B} & \cdots & A_{d_1d_1}\mlvec{B}
  \end{bmatrix}
\end{align}
where $A_{ij}$ is the $(i,j)$-th element of matrix $\mlvec{A}$.
As can be seen from the definition $\mlvec{A}\otimes\mlvec{B}\in\complex^{d_1d_2\times d_1 d_2}$. We also use the
symbol $\otimes$ for direct product of Hilbert spaces depending on the context.

By definition, the trace of $\mlvec{A}\otimes\mlvec{B}$
\begin{align}
  \tr\bigl(  \mlvec{A}\otimes\mlvec{B}\bigr) = \sum_{j=1}^{d_1} A_{jj} \tr(\mlvec{B}) =  \tr(\mlvec{A})\tr(\mlvec{B})
\end{align}

% =====================================================
\subsection{Quantum Information}
We now describe notions in quantum information using linear algebra. We start by
formulating quantum mechanics in the language of density matrices.

\textbf{Quatum states.}
A \emph{quantum state} with dimension $d$ is can be represented by a positive
semidefinite (PSD) Hermitian $\mlvec{\rho} \in\complex^{d\times d}$, the \emph{density
  matrix}. For any PSD Hermitian $\mlvec{\rho}$ with $\tr(\mlvec{\rho}) = 1$,
there is a correponding quantum state, and vice versa.

Quantum states with rank-$1$ density matrices are referred to as \emph{pure
  states}. For a pure state, its density matrix $\mlvec{\rho}$ allows a
eigen-decompostion $\mlvec{\rho} = \mlvec{v}\mlvec{v}^\dagger$ with $\mlvec{v}\in\complex^d$
being a $\ell_2$-unit vector.
$\mlvec{v}$ is referred to as the state vector representation of a
pure quantum state.

The  state space of a system composed of two subsystem with dimension
$d_1$ and $d_2$ has a dimension of $d_1 d_2$. And a density matrix for a such
composite system lies in $\complex^{d_1\times d_1}\otimes\complex^{d_2\times
  d_2}$. If a state can be expressed as a Kronecker product of density matrices,
we say it is a \emph{product state}. 
A simplistic example is $\mlvec{\rho}_{12} := \mlvec{\rho}_1\otimes \mlvec{\rho}_2$ with
$\mlvec{\rho}_1\in\complex^{d_1}$ and $\mlvec{\rho}_2\in\complex^{d_2}$.

Analogous to classical binary bits, the basic element for
quantum computers are \emph{qubits}. The state space of a single qubit is
$2$-dimensional, and the density matrix for a system composed of $n$ qubits lies
in $\otimes_{i=1}^n\complex^{2\times 2}$.

\textbf{Unitary gates.}
An operation over a quantum state is a linear map that is completely positive and
preserves the trace (See \cite{watrous2018theory} for a rigorous definition). Throughout this
paper we focus on \emph{unitary operators}. In the context of quantum circuit models,
operations are also referred to as \emph{gates}.

Under the density matrix representation, a unitary gate, denoted by
$\mlvec{U}\in\complex^{d\times d}$, transforms a state
$\mlvec{\rho}\in\complex^{d\times d}$ to $\mlvec{\rho}^\prime =
\mlvec{U}\mlvec{\rho} \mlvec{U}^\dagger$. The positive semidefiniteness and the
trace are preserved. For a pure state $\mlvec{\rho} =
\mlvec{v}\mlvec{v}^\dagger$, under the state vector representation, the unitary
gate $\mlvec{U}$ transforms $\mlvec{v}$ to $\mlvec{v}^\prime = \mlvec{U}\mlvec{v}$.

A set of unitary gates commonly used on a qubit are the \emph{Pauli
gates}: 
\begin{align}\label{eq:paulis}
   \mlvec{X} =
    \begin{bmatrix}
      0 & 1 \\
      1 & 0
    \end{bmatrix},\
  \mlvec{Y} =
    \begin{bmatrix}
            0 & -i \\
            i & 0 \\
          \end{bmatrix},\
  \mlvec{Z} =
    \begin{bmatrix}
      1 & 0 \\
      0 & -1 \\
    \end{bmatrix};
\end{align}
Note that the Pauli gates are both unitary and Hermitian.

We are especially interested in unitary gates parameterized by $\theta\in\real$
as $\exp(-i\theta\mlvec{H})$ for a Hermitian $\mlvec{H}$. We sometimes
refer to $\mlvec{H}$ as the \emph{Hamiltonian} that generates the parameterized quantum gate.
For any $\lambda\in\real$, shifting $\mlvec{H}$ by $\lambda\cdot\mlvec{I}$ does
not change the parameterized gate. To see this, for any $\mlvec{\rho}$, consider $\tilde{\mlvec{H}} =
\mlvec{H} + \lambda\mlvec{I}$:
\begin{align}
  {e^{-i\theta\tilde{\mlvec{H}}}}\mlvec{\rho}{e^{i\theta\tilde{\mlvec{H}}}}
  &= e^{-i\theta\lambda}\cdot{e^{-i\theta{\mlvec{H}}}}\mlvec{\rho}{e^{i\theta{\mlvec{H}}}}\cdot e^{i\theta\lambda}\\
  &= {e^{-i\theta{\mlvec{H}}}}\mlvec{\rho}{e^{i\theta{\mlvec{H}}}}.
\end{align}

\paragraph{Quantum measurements and observables} A measurement of quantum states
is specified by 
a set of matrices $\{\mlvec{M}_m\}\subset \complex^{d\times d}$ with
$\mlvec{M}_m^\dagger\mlvec{M}_m = \mlvec{I}_{d\times d}$, and
a set of outcomes $\{\lambda_m\}\subset \real$. 
Such measurements on the density matrix $\mlvec{\rho}$ yield the outcome
$\lambda_m$ with probability $\tr\bigl( \mlvec{M}_m^\dagger\mlvec{M}_m\mlvec\rho \bigr)$.
The probabilities are normalized
\begin{align}
  &\sum_m\tr\bigl( \mlvec{M}_m^\dagger\mlvec{M}_m\mlvec\rho \bigr)\\
  =&\tr\bigl( \sum_m\mlvec{M}_m^\dagger\mlvec{M}_m\mlvec\rho \bigr)\\
  =&\tr\bigl(\mlvec\rho \bigr) = 1.
\end{align}
In this paper we mainly focus on the expected outcome of a measurement. Let
$\mlvec{M}$ denote the \emph{observable} $\sum_m \lambda_m
\mlvec{M}_m^\dagger\mlvec{M}_m$, the expected value of the outcome is  
\begin{align}
  \sum_m\lambda_m\tr\bigl( \mlvec{M}_m^\dagger\mlvec{M}_m\mlvec\rho \bigr)
&  =\tr\bigl( \sum_m\lambda_m\mlvec{M}_m^\dagger\mlvec{M}_m\mlvec\rho \bigr)\\
&  =\tr\bigl( \mlvec{M} \mlvec\rho \bigr).
\end{align}

\setcounter{equation}{0}
\section{Proofs for Constructions}
\label{sec:app_constructions}
In this section, we provide the detailed proof for the existence of constructions appeared in
Section~\ref{sec:construction}.

We start by recalling the definitions of $\Phi_{l}^{(j)}(\cdot)$ and summarize
some useful facts about these linear maps in
Subsection~\ref{subsec:app_properties}, as well as the observables in the
Heisenberg picture.
In Subsection~\ref{subsec:app_construction_thm} we elaborate on our result on
the existence of hard datasets for $p$-parameter QNNs with linear independence, which leads to the
following theorem:
\constructthm*
In Subsection~\ref{subsec:app_construction_prop} we identify a family of 1-layer
QNNs with linear independence, which gives rise to Proposition~\ref{prop:onelayer}.
In addition, we instantiated concrete datasets to illustrate the generality of
our construction. 
\subsection{Linear Maps $\Phi_{l}^{(j)}(\cdot)$}
\label{subsec:app_properties}
Recall the definitions of  $\Phi_{l}^{(j)}$:
For all $l\in[p]$, define $\Phi_{l}^{(j)}$ such that for any Hermitian $\mlvec{A}$:
\begin{align}
  \Phi_{l}^{(0)}(\mlvec{A}) &= \frac{1}{2}(\mlvec{A} + \mlvec{H}_l \mlvec{A} \mlvec{H}_l)\\
  \Phi_{l}^{(1)}(\mlvec{A}) &= \frac{1}{2}(\mlvec{A} - \mlvec{H}_l \mlvec{A} \mlvec{H}_l)\\
  \Phi_{l}^{(2)}(\mlvec{A}) &= \frac{i}{2}[\mlvec{H}_l, \mlvec{A}]
\end{align}
The subscript $l$ will be dropped for general $\mlvec{H}$.

It can be easily verified that these mappings maps Hermitians to Hermitians, and
the traces of the output:
\begin{align}
  \tr\left(\Phi^{(0)}_l(\mlvec{A})\right) &= \tr\left(\frac{\mlvec{A} + \mlvec{H}_l\mlvec{A}\mlvec{H}_l}{2}\right) = \tr(\mlvec{A})\\
  \tr\left(\Phi^{(1)}_l(\mlvec{A})\right) &= \tr\left(\frac{\mlvec{A} - \mlvec{H}_l\mlvec{A}\mlvec{H}_l}{2}\right) = 0\\
  \tr\left(\Phi^{(2)}_l(\mlvec{A})\right) &= \tr\left(\frac{i[\mlvec{H}_l, \mlvec{A}]}{2}\right) = 0
\end{align}

\paragraph{Adjoint} The adjoints of the maps are:
\begin{align}
  \left(\Phi^{(0)}_l\right)^*(\mlvec{A}) &= \Phi^{(0)}_l(\mlvec{A}) \\
  \left(\Phi^{(1)}_l\right)^*(\mlvec{A}) &= \Phi^{(1)}_l(\mlvec{A}) \\
  \left(\Phi^{(2)}_l\right)^*(\mlvec{A}) &= -\Phi^{(2)}_l(\mlvec{A}) 
\end{align}
For all pairs of $j,k\in\{0,1,2\}$, we summarize the compostition of mappings
$(\Phi^{(j)}_l)^*\circ \Phi^{(k)}_l (\cdot)$  in Table~\ref{table:adjoints}, and
the inner products $\<\Phi^{(j)}_l(\mlvec{A}), \Phi^{(k)}_l(\mlvec{A})\>$ in Table~\ref{table:inner-product}.
\begin{table}[!htp]
  \caption{Composition $(\Phi^{(j)}_l)^*\circ \Phi^{(k)}_l (\cdot)$}
  \label{table:adjoints}
  \centering
  \begin{tabular}{|l|l|l|l|}
    \hline
    j\textbackslash{}k & 0 & 1 & 2 \\ \hline
    0 & $\Phi^{(0)}_l(\cdot)$ & $0$ & $0$ \\ \hline
    1 & $0$ & $\Phi^{(1)}_l(\cdot)$ & $\Phi^{(2)}_l(\cdot)$ \\ \hline
    2 & $0$ & $-\Phi^{(2)}_l(\cdot)$ & $\Phi^{(1)}_l(\cdot)$ \\ \hline
  \end{tabular}
\end{table}
\begin{table}[!htp]
  \caption{Inner product $\<\Phi^{(j)}_l(\mlvec{A}), \Phi^{(k)}_l (\mlvec{A})\>$}
  \label{table:inner-product}
  \centering
  \begin{tabular}{|l|l|l|l|}
    \hline
    j\textbackslash{}k & 0 & 1 & 2 \\ \hline
    0 & $\<\mlvec{A},\Phi^{(0)}_l(\mlvec{A})\>$ & $0$ & $0$ \\ \hline
    1 & $0$ & $\<\mlvec{A},\Phi^{(1)}_l(\mlvec{A})\>$ & $0$ \\ \hline
    2 & $0$ & $0$ & $\<\mlvec{A}, \Phi^{(1)}_l(\mlvec{A})\>$ \\ \hline
  \end{tabular}
\end{table}

All off-diagonal elements are zero in Table~\ref{table:inner-product}, implying
the orthogonality of $\{\Phi_{l}^{(j)}\}_{j=0,1,2}$. This
follows from the fact
\begin{align}
  \<\mlvec{A}, \Phi_l^{(2)}(\mlvec{A})\> = \frac{i}{2}(\tr(\mlvec{AHA}) - \tr(\mlvec{A}^2\mlvec{H})) = 0
\end{align}
We are now ready to prove the expansion of the observable in Heisenberg picture:
\begin{claim}[Observable in Heisenberg Picture]
  \label{claim:heisenbergpic}
  For quantum neural networks defined in Theorem~\ref{thm:construction}, the
  observable in Heisenberg picture
  $\mlvec{U}(\mlvec\theta)^\dagger\mlvec{M}\mlvec{U}(\mlvec{\theta})$ can be
  expressed as:
  \begin{align}
    \sum_{\mlvec{\xi}\in\{0,1,2\}^p}\Phi_{\mlvec{\xi}}(\mlvec{M})
    \prod_{l:\xi_{l}=1}\cos2\theta_{l}
    \prod_{l^\prime:\xi_{l^\prime}=2}\sin2\theta_{l^\prime}
  \end{align}
  where $\Phi_{\mlvec{\xi}}$ is defined as the following composed mapping:
  \begin{align}
    \Phi_1^{(\xi_1)}\circ \Phi_2^{(\xi_2)}\circ \cdots\circ
    \Phi_p^{(\xi_p)}
  \end{align}
\end{claim}
\begin{proof}
  For a two-level Hamiltonian $\mlvec{H}$ with eigenvalues $\pm1$, let
  $\mlvec{P}_+$ and $\mlvec{P}_-$ be projections into subspaces of $\complex^d$
  corresponding to eigenvalues $+1$ and $-1$:
  \begin{align}
    \mlvec{H} = \mlvec P_{+} - \mlvec P_{-},\ \text{and}\ \mlvec P_+  + \mlvec P_- = \mlvec{I};
  \end{align}

  For all $\theta\in\real$, the parameterized unitary 
  \begin{align}
    \exp(-i\theta\mlvec{H})
    &= e^{-i\theta}\mlvec P_+ + e^{i\theta} \mlvec P_- \\
    &= \cos\theta (\mlvec P_+  + \mlvec P_-) - \sin\theta i(\mlvec P_+ - \mlvec P_-) \\
    &= \cos\theta \mlvec I - i \sin\theta \mlvec H
  \end{align}

 By basic trigonometry, $e^{i\theta\mlvec{H}}{\mlvec{A}}e^{-i\theta\mlvec{H}}$
 can be expressed as
 \begin{align}
   \Phi^{(0)}(\mlvec{A}) + \Phi^{(1)}(\mlvec{A})\cos 2\theta + \Phi^{(2)}(\mlvec{A})\sin 2\theta
   \label{eqn:one-step-trig}
 \end{align}
 for any Hermitian $\mlvec{A}$. Claim~\ref{claim:heisenbergpic} then follows
 from sequantial application of Eqn.~(\ref{eqn:one-step-trig}).
\end{proof}

Under transformation 
$\mlvec\theta\mapsto\mlvec\theta + \mlvec{\alpha}$ for some
$\mlvec{\alpha}\in\real^p$, the linear maps transform as:
\begin{align}
  \begin{cases}
    \Phi_{l}^{(0)}(\cdot) &\rightarrow \Phi_{l}^{(0)}(\cdot)\\
    \Phi_{l}^{(1)}(\cdot) &\rightarrow \cos 2\alpha_l\Phi_{l}^{(1)}(\cdot) + \sin 2\alpha_l\Phi_{l}^{(2)}(\cdot)\\
    \Phi_{l}^{(2)}(\cdot) &\rightarrow -\sin 2\alpha_l\Phi_{l}^{(1)}(\cdot) + \cos 2\alpha_l\Phi_{l}^{(2)}(\cdot)
  \end{cases}
\end{align}
This is because
$e^{i(\theta+\alpha)\mlvec{H}}{\mlvec{A}}e^{-i(\theta+\alpha)\mlvec{H}}$ can be
expressed as
\begin{multline}
  \Phi^{(0)}(\mlvec{A}) 
  +(\cos 2\alpha\Phi^{(1)}(\mlvec{A}) + \sin 2\alpha\Phi^{(2)}(\mlvec{A}))\cos 2\theta\\
  +(-\sin 2\alpha\Phi^{(1)}(\mlvec{A}) + \cos 2\alpha\Phi^{(2)}(\mlvec{A}))\sin 2\theta
\end{multline}
As a consequence, each term of $f(\mlvec\rho;\mlvec\theta)$ remain
invariant under the joint transformation $\theta_l\mapsto\theta_l +
\frac{\pi}{2}$ and
\begin{align}
  \Phi^{(0)}_{l}(\cdot)\mapsto&\mlvec{H}_l\Phi^{(0)}_{l}(\cdot)\mlvec{H}_l =\Phi^{(0)}_{l}(\cdot)\\
  \Phi^{(1)}_{l}(\cdot)\mapsto&\mlvec{H}_l\Phi^{(1)}_{l}(\cdot)\mlvec{H}_l =-\Phi^{(1)}_{l}(\cdot)\\
  \Phi^{(2)}_{l}(\cdot)\mapsto&\mlvec{H}_l\Phi^{(2)}_{l}(\cdot)\mlvec{H}_l =-\Phi^{(2)}_{l}(\cdot)
\end{align}
\subsection{Proof for Lemma~\ref{lm:construction_sym} and \ref{lm:construction_breaksym}}
\label{subsec:app_construction_thm}
The construction for Theorem~\ref{thm:construction} consists of two steps. For
the first step, dataset $\S_0$ is constructed with $2^p$ local minima invariant
under the $\frac{\pi}{2}$ translational symmetry:
\begin{align}
  \theta_l \mapsto \theta_l + \frac{\pi}{2}.
\end{align}
Therefore the existence of a single local minimum $\mlvec{\theta}^
\star$ indicates a set of local minima $\Theta$.
For the second step, we construct a data set $\S_1$ to break the symmetry.
A combination of these two data set with proper scaling gives us a desired
dataset for Theorem~\ref{thm:construction}.

Here we provide the proof of Lemma~\ref{lm:construction_sym} and
Lemma~\ref{lm:construction_breaksym} in details.
% Statement of Lemma 3
\constructsymlm*
\begin{proof}
  It suffices to construct a dataset $\S_0 =
  \{(\mlvec\rho_k,y_k)\}_{k=1}^{m_0}$, such that (1) for all $k\in[p]$, $f_k(\mlvec{\theta}):=\<\mlvec\rho_k,
  \mlvec{M}(\mlvec\theta)\> - y_k$ is either symmetric or anti-symmetric under
  $\theta_l\mapsto\theta_l + \frac{\pi}{2}$ for all $l\in[p]$, and (2) 
  the intersection $\Theta$ of the set of roots $\Theta_k$ of $f_k(\mlvec{\theta})
  = 0$ is non-empty and contains at least one isolated point
  $\mlvec{\theta}^\star$. For such $\S_0$, $\mlvec{\theta}^\star$ is an isolated
  root of the non-negative loss function $L(\mlvec{\theta}; \S_0) =
  \sum_{k=1}^{m_0}f_k(\mlvec{\theta})^2$.

  For a concrete construction, consider $\{f_k(\mlvec\theta)\}_{k=1}^{m_0}$ such that
  $f_k(\mlvec\theta) = \propto \sin(2 \sum_{l=1}^p\eta^{(k)}_l(\theta_l -
  \theta^\star_l))$ for a set of vectors
  $\{\mlvec\eta^{(k)}\}_{k=1}^{m_0}\subseteq\{-1,0,1\}^{p}_{l=1}$ that spans
  $\real^p$, and arbitrary $\mlvec{\theta}^\star\in\real^p$.

  The translational symmetry holds due to the periodicity of $\sin$-functions;
  the loss function has a minima at the vanishing point $\mlvec\theta^\star$.

To see that $\mlvec\theta^\star$ is indeed a isolated minima (that there
exists a neighbourhood within which the loss function vanishes only at
$\mlvec{\theta}^\star$), consider
$\N(\mlvec{\theta}^\star) := \{\mlvec{\theta}\in\real^p: \forall k\in[m_0],
|(\mlvec{\eta}^{(k)})^T(\mlvec{\theta} - \mlvec\theta^\star)| < \frac{\pi}{4}\}$.
Conditioned on $\mlvec\theta \in \N(\mlvec\theta^\star)$,
\begin{align}
  L(\mlvec\theta; \S_0) = 0 &\implies\\
  &\forall k\in[m_0], \sum_{l=1}^{p}\eta^{(k)}_{l}(\theta_l - \theta_l^\star) = 0
\end{align}
which then implies $\mlvec\theta = \mlvec\theta^\star$ since $\{\mlvec\eta^{(l)}_k\}_{l=1}^p$ spans $\real^p$.

The existence of such dataset $\S_0$ follows from the linear independence of
operators for the QNN: for any $k\in[m_0]$, the solution set to the following linear system for Hermitian $\mlvec{D}_k\in
\complex^{d\times d}$ is non-empty:
\begin{align}
  \begin{cases}
    &\<\mlvec{D}_k, \mlvec{I}\> = 0, \\ 
    &\<\mlvec{D}_k, \Phi_{\mlvec{\xi}}(\mlvec{M})\> = \hat{f}_{\mlvec{\xi},k},\ \forall \mlvec{\xi}\neq\mlvec{0}. 
  \end{cases}
\end{align}
where $\mlvec{I}$ is the $d$-dimensional identity, and $\hat{f}_{\mlvec{\xi},k}$ denotes the coefficient corresponding to the
term $\prod_{l:\xi_l=1}\cos2\theta_l\prod_{l^\prime:\xi_{l^\prime}=2}\sin2\theta_{l^\prime}$ in 
$\sin(2\sum_{l=1}^p\eta^{(k)}_l (\theta_l-\theta^\star_l))$.

As $\tr(\Phi_{\mlvec{\xi}}(\mlvec{M})) = 0$, $\mlvec{I}$ is orthogonal to all
$\Phi_{\mlvec{\xi}}(\mlvec{M})$. Therefore the constraint set
$\{\Phi_{\mlvec{\xi}}(\mlvec{M})\}_{\mlvec{\xi}\neq\mlvec{0}}\cup\{\mlvec{I}\}$ is linear
independent. As a result, the linear system is guaranteed to have a set of
solution $\{\mlvec{D}_k\}_{k=1}^{m_0}$.

Given the solution $\{\mlvec{D}_k\}_{k=1}^{m_0}$, the dataset can be constructed
as:
\begin{align}
  \mlvec\rho_k:=\frac{1}{d}\mlvec{I} + \kappa \mlvec{D}_k,\quad
  y_k = \tr(\mlvec{\rho}_k\Phi_{\mlvec{0}}(\mlvec{M}))
\end{align}
for all $k$, with $\kappa$ be the largest positive real number such that 
$\frac{1}{d}\mlvec{I} + \kappa \mlvec{D}_k$ is positive semidefinite.

It can be verified by elementary calculation that such dataset yields a loss
function $L(\mlvec\theta;\S_0) = \kappa\sum_{k=1}^{m_0}\sin(2(\mlvec{\eta}^{(k)})^T(\mlvec{\theta}-\mlvec{\theta}^\star))$.
\end{proof}
\constructbreaksymlm*
\begin{proof}
  Rewrite the loss function induced by $\S_1$ as:
  \begin{align}
    L(\mlvec\theta; \S_1) &= - \frac{2}{m_1}\sum_{k=1}^{m_1}\<y_k\mlvec\rho_k, \mlvec{M}(\mlvec\theta)\>\label{line:term1}\\
                          &+ \frac{1}{m_1}\sum_{k=1}^{m_1}\(\<\mlvec\rho_k, \mlvec{M}(\mlvec\theta)\>\)^2\label{line:term2}\\
                          &+ \frac{1}{m_1}\sum_{k=1}^{m_1}y_k^2
  \end{align}
  For any positive $\epsilon$, consider the following joint scaling of $\mlvec{\rho}_k$ and $y_k$:
  \begin{align}
    \begin{cases}
      \mlvec\rho_k &\mapsto \epsilon\mlvec\rho_k + (1 - \epsilon)\frac{1}{d} \mlvec{I} \\
      y_k &\mapsto \frac{1}{\epsilon}y_k
    \end{cases}
  \end{align}
  Under such scaling, for arbitrary $\epsilon$, term~(\ref{line:term1}) remains the
  same; the term~(\ref{line:term2}) can be arbitrarily suppressed by choosing
  sufficiently small $\epsilon$.
  Therefore it suffices to consider the first term $L^\prime(\cdot; \S_1) :=
  -\frac{2}{m_1}\sum_{k=1}^{m_1}y_k\<\mlvec\rho_k,\mlvec{M}(\mlvec{\theta})\>$.

  Without loss of generality, assume $\mlvec{\theta}^\star = \mlvec{0}$.
  Consider a dataset $\S_1$ such that
  $L^\prime(\mlvec\theta;\S_1) \propto -\sum_{k=1}^{m_1} \cos(2
  (\tilde{\mlvec{\eta}}^{(k)})^T\mlvec{\theta})$ for a set
  $\{\tilde{\mlvec{\eta}}^{(k)}\}_{k=1}^{m_1}\subset\{0,1\}^p$ that spans $\{0,1\}^p$. 

  For any $\mlvec{\zeta} \in\{0,1\}^p$,
  \begin{align}
    &L^\prime(\mlvec{\theta}^\star + \frac{\pi}{2}\mlvec{\zeta}; \S_1)\\
    \propto & -\sum_{k=1}^{m_1} \cos(
              2(\tilde{\mlvec{\eta}}^{(k)})^T\mlvec{\theta}^\star + 
              \<\tilde{\mlvec{\eta}}^{(k)}, \mlvec{\zeta}\>\pi)\\
    \propto &  -\sum_{k=1}^{m_1} (-1)^
              {\<\tilde{\mlvec{\eta}}^{(k)}, \mlvec{\zeta}\>}
  \end{align}

  The fact that $\{\tilde{\mlvec\eta}^{(k)}\}_{k=1}^{m_1}$ spans $\{0,1\}^p$
  indicate that the solution to
  \begin{align}
    \forall k\in[m_1],\ 
    \<\tilde{\mlvec{\eta}}^{(k)}, \mlvec{\zeta}\> = 0\ (\mathit{mod}\ 2)
  \end{align}
  is unique. Therefore such $\S_1$ breaks the $\pi/2$-translational symmetry as
  required. Similar to the proof of
  Lemma~\ref{lm:construction_sym}, the existence of such dataset follows from the
  linear independence of the operators.

  For a concrete construction, let $\S_1$ be a dataset such that
$L^\prime(\mlvec\theta;\S_1) = -\frac{2}{m_1}\sum_{k=1}^{m_1}\<y_k\mlvec\rho_k,\mlvec{M}(\theta)\> =
-c\sum_{l=1}^p\cos(\theta_l)$.  Due to the flexible scaling of the labels
$\{y_k\}_{k=1}^{m_1}$, such $\S_1$ can be found for arbitrary $c$.

Let $\B_r(\mlvec\theta)$ denote the closed
$\ell_2$-ball centered at $\mlvec{\theta}$ with radius $r$, and let
$\partial\B_r$ denote its boundary.
For the local minimum $\mlvec\theta^\star$ of $L(\mlvec{\theta};\S_0)$, let $(r,
L_0)$ be a pair of real numbers such that
$\inf_{\mlvec{\theta}\in\partial\B_r(\mlvec\theta^\star)}L(\mlvec{\theta};
\S_0) > L_0$ for some positive real number $L_0$.

The requirements in Lemma~\ref{lm:construction_breaksym} is then met by choosing
$c<\frac{L_0}{2pr^2}$. To see this we will make use of the estimation $1 - \frac{1}{2}
\theta_l^2 \leq \cos\theta_l \leq 1$. 
For any $\mlvec\theta^\prime\in\Theta$, the loss function
$L(\mlvec\theta;\S_0) + L^\prime(\mlvec\theta;\S_1)$ evaluated at all points on the 
boundary of $B_r(\theta^\prime)$ is at least $\frac{L_0}{2}$ larger than
$L(\mlvec\theta^\prime,\S_0) + L^\prime(\mlvec{\theta}^\prime, \S_1)$. Therefore
the second requirement in Lemma~\ref{lm:construction_breaksym} is met.

For the first requirement, we have
that (1) $L(\mlvec{\theta^\star};\S_0) + L^\prime(\mlvec{\theta}^\star;\S_0) =
-pc$, and (2) for all $\mlvec\theta^\prime\in\Theta/\{\mlvec\theta^\star\}$, for
all $\mlvec\theta\in\B_r(\mlvec\theta^\prime)$, $L(\mlvec{\theta};\S_0) + L^\prime(\mlvec{\theta};\S_0) >
 0 -(p-1)c + (1-\frac{1}{2}r^2)c = -pc + (2 - \frac{1}{r^2}) c$. The suboptimality
gap is therefore at least $c(2 - \frac{1}{2}r^2)$.
\end{proof}
\paragraph{Remarks.} In the proof for Lemma~\ref{lm:construction_sym} and
\ref{lm:construction_breaksym}, we made use of specific forms of $\sin$- and
$\cos$-functions for the clarity of proof. However, the linear independence of the
operators allow us to construct loss function beyond these specific forms, as will be
made clear in Example~\ref{example:address_minima_position} and \ref{example:address_decomposability}.

\subsection{Proof for Proposition~\ref{prop:onelayer} and Concrete Constructions}
\label{subsec:app_construction_prop}
\onelayerprop*
\begin{proof}
Proposition~\ref{prop:onelayer} follows directly from the orthogonality of the
operators. For any $\mlvec{\xi}\in\{0,1,2\}^p$, the operator
$\Phi_{\mlvec{\xi}}(\mlvec{M})$ can be expressed in the tensor product form $\otimes_{l=1}^p\tilde\Phi_{l}^{(\xi_l)}(\mlvec{m}_l)$.
where $\tilde\Phi_{l}^{(j)}$ are linear maps associated with $\mlvec{h}_l$.
For any $\mlvec{\xi}$ and $\mlvec{\xi}^\prime \in \{0,1,2\}^p / {\mlvec{0}}$:
\begin{align}
    & \<\Phi_{\mlvec{\xi}}(\mlvec{M}), \Phi_{\mlvec{\xi}^\prime}(\mlvec{M})\>\\
  = & \prod_{l=1}^p\<\tilde\Phi_{l}^{(\xi_l)}(\mlvec{m}_l), \tilde\Phi_{l}^{(\xi^\prime_l)}(\mlvec{m}_l)\>\\
  = & \prod_{l=1}^p\|\tilde\Phi_{l}^{(\xi_l)}(\mlvec{m}_l)\|_F^2\delta_{\xi_l, \xi^\prime_l}
\end{align}
Therefore $\{\Phi_{\mlvec\xi}(\mlvec{M})\}$ forms a linear independent set.
\end{proof}
We now move on to concrete constructions of datasets for QNNs defined in
Proposition~\ref{prop:onelayer}. To facilitate narrative, assume $\mlvec{M}
= \mlvec{M}_0^{\otimes p}$, and $\mlvec{h}_l = \mlvec{H}_0$.

For $j=0,1,2$, define $\mlvec{D}^{(j)}$ as $\Phi^{(j)}(\mlvec{M}_0)$.
And let $\mlvec{\rho}^{(j)}$ be a proper linear combination of $\mlvec{D}^{(j)}$
and $\mlvec{I}$ that is positive semidefinite and trace-$1$.

For $l\in[p]$, define
\begin{align}
  \mlvec{\rho}_{0,l} &= \bigl(\otimes_{r=1}^{l-1}\mlvec \rho^{(0)}\bigr) \otimes \mlvec \rho^{(1)} \otimes \bigl(\otimes_{r=l+1}^p \mlvec \rho^{(0)}\bigr)\\
  \mlvec{\rho}_{1,l} &= \bigl(\otimes_{r=1}^{l-1}\mlvec \rho^{(0)}\bigr) \otimes \mlvec \rho^{(2)} \otimes \bigl(\otimes_{r=l+1}^p \mlvec \rho^{(0)}\bigr)\\
  y_{0,l} &= 0\\
  y_{1,l} &= \tr(\rho^{(0)}\Phi^{(0)}(\mlvec{M}_0))^{p-1}\tr(\rho^{(1)}\Phi^{(1)}(\mlvec{M}_0))
\end{align}

Let $\S_0 = \{(\mlvec{\rho}_{0,l}, y_{0,l})\}_{l=1}^p$ be the dataset with
$\frac{\pi}{2}$-translational symmetry, with the resulting loss function
proportional to $\sum_{l=1}^p\sin\theta_l^2$; Let $\S_1 = \{(\mlvec{\rho}_{1,l},
y_{1,l})\}_{l=1}^p$ be the dataset that breaks the symmetry, with the loss
function proportional to $\sum_{l=1}^p(\cos\theta_l - 1)^2$.

\begin{example}
  \label{example:base}
  For a concrete example, consider $\mlvec{M}_0 = \mlvec{Y}+ \mlvec{I}$, and
  $\mlvec{H}_0 =\mlvec{Z}$. Choose
  $\mlvec{\rho}^{(0)} = \frac{1}{2}(\mlvec{Z} + \mlvec{I})$, 
  $\mlvec{\rho}^{(1)} = \frac{1}{2}(\mlvec{X} + \mlvec{I})$, 
  $\mlvec{\rho}^{(2)} = \frac{1}{2}(\mlvec{Y} + \mlvec{I})$.
  Construct the dataset $\S_0$ and $\S_1$ as described above, and let the
  dataset $\S$ be the combination of $\S_0$ and $\S_1$ with reweighting factor
  $4:1$. The loss function takes the form:
  \begin{align}
    \frac{1}{2p}\sum_{l=1}^p \sin^2(2\theta_l) + \frac{1}{4}(\cos 2\theta_l - 1)^2
  \end{align}
\end{example}

For clarity, the construction in Example~\ref{example:base} was purposefully
designed with two limitations. First of all, the construction has a fixed global minima at
$\mlvec\theta^\star=\mlvec{0}$. Also, the
loss function of the construction in Example~\ref{example:base} can be 
decomposed into $p$ single-parameter functions. Therefore the training
problem can be solved by greedily optimizing each of the coordinate $\theta_l$.

To address the first limitation, we propose the following example:
\begin{example}
  \label{example:address_minima_position}
  For
  \begin{align}
    \mlvec{\rho}_{0,l} &= \bigl(\otimes_{r=1}^{l-1}\mlvec \rho^{(0)}\bigr) \otimes \mlvec \rho^{(1)} \otimes \bigl(\otimes_{r=l+1}^p \mlvec \rho^{(0)}\bigr)\\
    \mlvec{\rho}_{1,l} &= \bigl(\otimes_{r=1}^{l-1}\mlvec \rho^{(0)}\bigr) \otimes \mlvec \rho^{(2)} \otimes \bigl(\otimes_{r=l+1}^p \mlvec \rho^{(0)}\bigr)\\
    y_{0,l} &= \sin(\frac{\pi}{50}),\quad y_{1,l} = \cos(\frac{\pi}{50})
  \end{align}
  Construct dataset $\S_0$, $\S_1$ as:
  \begin{align}
    \S_0 = \{(\mlvec{\rho}_{0,l}, y_{0,l})\}_{l=1}^p,\quad
    \S_1 = \{(\mlvec{\rho}_{1,l}, y_{1,l})\}_{l=1}^p
  \end{align}
  and let the
  dataset $\S$ be the combination of $\S_0$ and $\S_1$ with reweighting factor
  $4:1$. The loss function takes the form:
  \begin{multline}
    \frac{1}{2p}\sum_{l=1}^p\big((\sin2\theta_l - \sin\frac{\pi}{50})^2 \\+ \frac{1}{4}(\cos2\theta_l - \cos\frac{\pi}{50})^2 \big).
  \end{multline}
\end{example}
The resulting loss function has a local minimum at
$(\frac{\pi}{100},\cdots,\frac{\pi}{100})^T$.  

To address the decomposability issue, consider the following example:
\begin{example}
  \label{example:address_decomposability}
  Let $\S_0, \S_1$ be as defined in
  Example~\ref{example:address_minima_position}.
  Let $\mlvec{\rho}_{2,l,k}$ denote a product state with $\mlvec{\rho}^{(1)}$
  for the $k$-th and $l$-th qubits, and $\mlvec{\rho}^{(0)}$ for the rest.
  The loss function for training sample $(\mlvec{\rho}_{2,l,k}, \cos^2\frac{\pi}{50})$ is
  $(\cos 2\theta_l\cos 2\theta_k - \cos^2\frac{\pi}{50})^2$. Combining this additional term with $\S_0$
  and $\S_1$ gives rise to non-decomposable loss functions that cannot be solved
  by optimizing each coordinate independently. 
\end{example}

As will be seen in \supref{sec:app_num}, the construction of Example~\ref{example:address_minima_position} and
\ref{example:address_decomposability} indicates that our construction method is
general, and can lead to instances that are hard to optimize
with gradient-based methods and do not admit other trivial optimization methods.

\setcounter{equation}{0}
\section{Proof for Typical QNNs with Linear Dependence}
\label{sec:app_linear_indep}
\newcommand{\PhiXiM}[1][]{\Phi_{\mlvec{\xi}
    \ifx\\#1\\\else_{#1}\fi}
  (\mlvec{M})}
\newcommand{\PhiXiMprime}[1][]{\Phi_{\mlvec{\xi}^\prime
    \ifx\\#1\\\else_{#1}\fi}
  (\mlvec{M})}
\newcommand{\Phimap}[2]{\Phi^{#1}_{#2}}

In this section, we show that typical under-parameterized QNNs are with linear
independence. To that end, we consider random $d$-dimensional $p$-parameter QNNs sampled with respect to the
following measure: 

Let $\mlvec{H}$ be a $d$-dimension Hermitian such that $\tr(\mlvec{H}) = 0$ and
$\mlvec{H}^2 = \mlvec{I}$. The circuit for our random QNN is:
\begin{align}
  \mlvec{U}(\mlvec{\theta}) =
  e^{-i\theta_p\mlvec{W}_p\mlvec{H}\mlvec{W}_p^\dagger}
  \cdots
  e^{-i\theta_1\mlvec{W}_1\mlvec{H}\mlvec{W}_1^\dagger}
  \label{eqn:app_random_model}
\end{align}
with $\{\mlvec{W}_l\}_{l=1}^p$ independently sampled with respect to the Haar measure on the $d$-dimensional unitary
group $U(d)$.

Following from the fact $e^{-i\theta\mlvec{W}\mlvec{H}\mlvec{W}^\dagger} =
\mlvec{W} e^{-i\theta\mlvec{H}} \mlvec{W}^\dagger$ for Hermitian $\mlvec{H}$ and
unitary $\mlvec{W}$, we can rewrite Eqn.~(\ref{eqn:app_random_model}) as:
\begin{multline}
  \mlvec{U}(\mlvec{\theta}) = \mlvec{W}_p
  e^{-i\theta_p\mlvec{H}} \big(\mlvec{W}_p^\dagger\mlvec{W}_{p-1}\big) \\
  \cdots\
  \big(\mlvec{W}_2^\dagger \mlvec{W}_1\big) e^{-i\theta_1\mlvec{H}} \mlvec{W}_1^\dagger
\end{multline}
Due to the left (and right) invariance of the Haar measure, up to a unitary transformation, the random model in Eqn.~(\ref{eqn:app_random_model}) is
equivalent to a circuit with $p$ interleaving parameterized gate
$\{e^{-i\theta_l\mlvec{H}}\}_{l=1}^p$ and unitary $\{\tilde{\mlvec W}_l\}_{l=1}^p$
randomly sampled with respect to the Haar measure:
\begin{align}
  \mlvec{U}(\mlvec{\theta}) = \tilde{\mlvec{W}}_pe^{-i\theta_p\mlvec{H}}\tilde{\mlvec{W}}_{p-1}\cdots\tilde{\mlvec{W}}_1e^{-i\theta_1\mlvec{H}}
  \label{eqn:app_random_model1}
\end{align}
This interleaving nature of fixed and parameterized gates is shared by existing
designs of QNNs, and any $p$-parameter QNN generated by two-level Hamiltonians
can be expressed in Eqn.~(\ref{eqn:app_random_model1}). Morever, applying
polynomially many random 2-qubit gates on random pairs of qubits generates a
distribution over gates that approximates the Haar measure up to the 4-th
moments \cite{brandao2016local}, which is what we require in the proof in this section.

In the rest of this section, we provide detailed proof for Theorem~\ref{thm:random_qnn}.
\randomthm*

\subsection{Proof of Theorem~\ref{thm:random_qnn}}
\label{subsec:app_random}
\begin{proof}
  
Let $\Xi=\{0,1,2\}^p / \{\mlvec{0}\}$ denote the set of all $\mlvec\xi \in \{0,1,2\}^p$ except for $\mlvec\xi =
(0,\cdots,0)^T$. Our goal is to show that
$\{\Phi_{\mlvec{\xi}}\}_{\xi\in\Xi}$ is linearly independent with high probability.

To show the linear independence of $\{\Phi_{\mlvec{\xi}}\}$, it suffices to show
that its Gram matrix is positive semidefinite (Theorem 7.2.10 in \citet{horn2013matrix}).
The Gram matrix $\mlvec{G}$ is defined such that the $(\mlvec\xi, \mlvec\xi^\prime)$-element $G_{\mlvec{\xi},\mlvec{\xi}^\prime}:=\<\Phi_{\xi}(\mlvec{M}),
\Phi_{\mlvec\xi^\prime}(\mlvec{M})\>$, for all pairs of $\mlvec\xi,\mlvec\xi^\prime \in\Xi$.

By the Gershgorin circle theorem \cite{golub1996matrix}, it
suffices to show that with high probability
\begin{multline}
  \<\Phi_{\mlvec{\xi}}(\mlvec{M}), \Phi_{\mlvec{\xi}}(\mlvec{M})\> >\\
  \sum_{\mlvec\xi^\prime\in\Xi, \mlvec\xi^\prime \neq \mlvec\xi} 
  |\<\Phi_{\mlvec{\xi}}(\mlvec{M}), \Phi_{\mlvec{\xi^\prime}}(\mlvec{M})\>|
\end{multline}
for all $\mlvec\xi \in \Xi$.

Using the Chebyshev inequality with the moment estimations in
Lemma~\ref{lm:exp_and_var}, we have:
\begin{align}
  \Pr\left[\tr(\Phi_{\mlvec\xi}(\mlvec{M})^2) < \frac{2}{3}\frac{\tr(\mlvec M^2)}{2^p} \right] \leq O(d^{-1})
\end{align}
and 
\begin{multline}
  \Pr\big[|\tr(\Phi_{\mlvec\xi}(\mlvec{M}) \Phi_{\mlvec\xi^\prime}(\mlvec{M}))|
    \\>
    \frac{1}{3\cdot 3^p}\frac{\tr(\mlvec M^2)}{2^p} \big] \leq O(3^pd^{-1})
\end{multline}
for $\mlvec\xi \neq \mlvec{\xi}^\prime$.
Combined with the union bound, 
we can show that:
\begin{multline}
  \Pr\big(
  \exists \mlvec{\xi}\in\Xi:
  \<\Phi_{\mlvec{\xi}}(\mlvec{M}) \Phi_{\mlvec{\xi}}(\mlvec{M})\>\leq\\
  \sum_{\mlvec{\xi}^\prime\in\Xi:\mlvec{\xi}^\prime\neq\mlvec{\xi}}
  |\<\Phi_{\mlvec{\xi}}(\mlvec{M}) \Phi_{\mlvec{\xi}^\prime}(\mlvec{M})\>|
  \big)
  \leq O(e^{p}d^{-1})
\end{multline}
\end{proof}

\subsection{Moments}
\begin{lemma}[Expectations and variances]\label{lm:exp_and_var}
  Consider the set of operators $\{\Phi_{\mlvec\xi}(\mlvec{M})\}$ of random
  $d$-dimensional $p$-parameter QNNs defined in
  Eqn.~(\ref{eqn:app_random_model}). The expectations of diagonal and
  off-diagonal terms of the associated Gram matrix are:
  \begin{align}
    \EXP[\tr(\Phi_{\mlvec\xi}(\mlvec{M})\Phi_{\mlvec\xi}(\mlvec{M}))] &= \frac{\tr(\mlvec{M}^2)}{2^p}(1 + O(pd^{-2}))\\
    \EXP[\tr(\Phi_{\mlvec\xi}(\mlvec{M})\Phi_{\mlvec\xi^\prime}(\mlvec{M}))] &= 0\\
  \end{align}
  for all $\mlvec{\xi}\in\Xi$ and $\mlvec{\xi}^\prime \neq \mlvec{\xi}$. The
  variances are:
  \begin{align}
    \mathbb{V}[\tr(\Phi_{\mlvec\xi}(\mlvec{M})\Phi_{\mlvec\xi^\prime}(\mlvec{M}))] &=  \frac{\tr(\mlvec{M}^2)^2}{4^p}O(d^{-1})
  \end{align}
  for all $\mlvec{\xi},\mlvec{\xi}^\prime\in\Xi$.
\end{lemma}

\begin{proof}
  Throughout the proof, we will use $\EXP_l[\cdot]$ to denote expectation with
  respect to $\mlvec{H}_l$, and use $\EXP_{l:p}[\cdot]$ to denote integral over the product measure over $\mlvec
  W_l, \cdots, \mlvec{W}_p$. The subscripts will be dropped when there is no confusion.

  We start by showing some basic (in-)equalities using the formula for integrals
  over the Haar measure~\cite{puchala2011symbolic}. 
  Recall that $\mlvec{H}_l = \mlvec{W}_l\mlvec{H}\mlvec{W}_l^\dagger$.
  Let $\mlvec{A}$ and $\mlvec{B}$ be two Hermitian matrices such that $\tr(\mlvec{A}) = \tr(\mlvec{B})
  = 0$. For integrals where $\mlvec{H}_l$ appears once:
  \begin{align}
    \EXP_l[\tr(\mlvec{AH}_l)] &= \frac{1}{d}\tr(\mlvec{A})\tr(\mlvec{H}) = 0
  \end{align}
  For integrals where $\mlvec{H}_l$ appears twice, we have:
  \begin{align}
     &\EXP_l[\tr(\mlvec{AH}_l)\tr(\mlvec{BH}_l)]\\
    =&\frac{\tr(\mlvec{AB})\tr(\mlvec{H}^2)+ \tr(\mlvec{A})\tr(\mlvec{B})\tr(\mlvec{H})^2}{d^2-1}\\
    -&\frac{\tr(\mlvec{AB})\tr(\mlvec{H})^2+ \tr(\mlvec{A})\tr(\mlvec{B})\tr(\mlvec{H}^2)}{d(d^2-1)}\\
    =&\frac{1}{d - d^{-1}}\tr(\mlvec{AB})
  \end{align}
  \begin{align}
    &\EXP_l[\tr(\mlvec{AH}_l\mlvec{BH}_l)]\\
    =&\frac{\tr(\mlvec{AB})\tr(\mlvec{H})^2+ \tr(\mlvec{A})\tr(\mlvec{B})\tr(\mlvec{H}^2)}{d^2-1}\\
    -&\frac{\tr(\mlvec{AB})\tr(\mlvec{H}^2)+ \tr(\mlvec{A})\tr(\mlvec{B})\tr(\mlvec{H})^2}{d(d^2-1)}\\
    =&-\frac{1}{d^2 - 1}\tr(\mlvec{AB})
  \end{align}
  For integrals with $\mlvec{H}_l$ appearing $4$ times, we use the
  following estimation   \begin{align}
     &|\EXP[\tr(\mlvec{A}\mlvec{H}_l\mlvec{A}\mlvec{H}_l)
           \tr(\mlvec{B}\mlvec{H}_l\mlvec{B}\mlvec{H}_l)]|\\
    = & O(d^{-4})\max\{|\tr(\mlvec{H}^2)|^2, |\tr(\mlvec{H}^4)|\}\\
                           \cdot &\max\{|\tr(\mlvec{A}^2\mlvec{B}^2)|,|\tr(\mlvec{A}\mlvec{B}\mlvec{A}\mlvec{B})|,\\
                                 &\quad\quad
                    |\tr(\mlvec{A}\mlvec{B})^2|, |\tr(\mlvec{A}^2)\tr(\mlvec{B}^2)|\}\\
  = & O(d^{-2}) \tr(\mlvec{A}^2)\tr(\mlvec{B}^2)
\end{align}
Here the first relation follows from formula on Equation~(3) in
\cite{puchala2011symbolic} and expressions in Sec. 6 of \cite{collins2006integration});
and the second relation follows from matrix Cauchy-Schwarz inequalities and
trace inequalities from \cite{petz1988characterizations,petz1994survey}.

Similarly we have:
  $|\EXP[\tr(\mlvec{A}\mlvec{H}_l\mlvec{B}\mlvec{H}_l)
    \tr(\mlvec{A}\mlvec{H}_l\mlvec{B}\mlvec{H}_l)]|
  =  O(d^{-2}) \tr(\mlvec{A}^2)\tr(\mlvec{B}^2)$.
\paragraph{First moments}
We start by calculating the first moments of
$\<\Phi_{\mlvec{\xi}}(\mlvec{M}),\Phi_{\mlvec{\xi}}(\mlvec{M})\>$
and
$\<\Phi_{\mlvec{\xi}}(\mlvec{M}), \Phi_{\mlvec{\xi^\prime}}(\mlvec{M})\>$.
\begin{align}
   &\EXP\<\Phi_{\mlvec{\xi}}(\mlvec{M}),\Phi_{\mlvec{\xi}^\prime}(\mlvec{M})\>\label{ln:app_expectation1}\\
  =&\EXP\<\Phi_{1}^{(\xi_1)}\circ\Phi_{\mlvec{\xi}_{2:p}}(\mlvec{M}),\Phi_{1}^{(\xi^\prime_1)}\circ\Phi_{\mlvec{\xi}^\prime_{2:p}}(\mlvec{M})\>\\
  =&\EXP\<\Phi_{\mlvec{\xi}_{2:p}}(\mlvec{M}),(\Phi_{1}^{(\xi_1)})^*\circ\Phi_{1}^{(\xi^\prime_1)}\circ\Phi_{\mlvec{\xi}^\prime_{2:p}}(\mlvec{M})\>
\end{align}

By the basic results on the adjoints in \supref{subsec:app_properties},
$(\Phi_{l}^{(\xi_l)})^*\circ\Phi_{l}^{(\xi^\prime_l)}(\cdot)
=\Phi_{l}^{(0)}(\cdot)$ (or $\Phi_{l}^{(1)}(\cdot)$) if $\xi_l = \xi_l^\prime=0$
(or $\xi_l = \xi_l^\prime\neq 0$). In the case $\xi_l \neq \xi_l^\prime$, if
$\xi_l$ and $\xi_l^\prime$ are both in $\{1,2\}$,
$(\Phi_{l}^{(\xi_l)})^*\circ\Phi_{l}^{(\xi^\prime_l)}(\cdot)
=\pm\Phi_{l}^{(2)}(\cdot)$. Otherwise
$(\Phi_{l}^{(\xi_l)})^*\circ\Phi_{l}^{(\xi^\prime_l)}(\cdot)=0$. We treat these
three cases separately.
\paragraph{Case 1: $\xi_1 = \xi_1^\prime$.}
  \begin{align}
   &(\ref{ln:app_expectation1})\\
  =&\EXP\<\Phi_{\mlvec{\xi}_{2:p}}(\mlvec{M}),\Phi_{1}^{(0/1)}\circ\Phi_{\mlvec{\xi}_{2:p}^\prime}(\mlvec{M})\>\\
  =&\frac{1}{2}\EXP_{2:p}\<\Phi_{\mlvec{\xi}_{2:p}}(\mlvec{M}),
     \Phi_{\mlvec{\xi}_{2:p}^\prime}(\mlvec{M})\>\\
\pm&\frac{1}{2}\EXP\<\Phi_{\mlvec{\xi}_{2:p}}(\mlvec{M}),
       \mlvec{H}_1\Phi_{\mlvec{\xi}_{2:p}^\prime}(\mlvec{M}) \mlvec{H}_1\>\\
  =&\frac{1}{2}(1 \mp \frac{1}{d^2-1})\EXP_{2:p}\<\Phi_{\mlvec{\xi}_{2:p}}(\mlvec{M}), \Phi_{\mlvec{\xi}_{2:p}^\prime}(\mlvec{M})\> 
\end{align}
The sign of the second term depends on whether $\xi_1$ is $0$ or $1,2$.
Therefore the first moment of the Frobenius norm of $\Phi_{\mlvec\xi}(\mlvec{M})$:
\begin{align}
  &\EXP\< \Phi_{\mlvec\xi}(\mlvec{M}), \Phi_{\mlvec\xi}(\mlvec{M})\>\\
  =& \frac{1}{2}(1\pm\frac{1}{d^2 - 1})
     \EXP_{2:p}
     \<\Phi_{\mlvec\xi_{2:p}}(\mlvec{M}), \Phi_{\mlvec\xi_{2:p}}(\mlvec{M})\>\\
  =& (1+O(p d^{-2})) \frac{\tr(\mlvec{M}^2)}{2^p}
\end{align}
\paragraph{Case 2: $\xi_1,\xi_1^\prime\in\{1,2\}$ and $\xi_1\neq \xi_1^\prime$.}
\begin{align}
  &(\ref{ln:app_expectation1})\\
  =&\pm\EXP\<\Phi_{\mlvec{\xi}_{2:p}}(\mlvec{M}),\Phi_{1}^{(2)}\circ\Phi_{\mlvec{\xi}_{2:p}^\prime}(\mlvec{M})\>\\
  =&\pm\frac{i}{2}\EXP_{2:p}\<\Phi_{\mlvec{\xi}_{2:p}}(\mlvec{M}),\\
   &\qquad\mlvec{H}_1\Phi_{\mlvec{\xi}_{2:p}^\prime}(\mlvec{M})
    -\Phi_{\mlvec{\xi}_{2:p}^\prime}(\mlvec{M})\mlvec{H}_1
     \>\\
  =&\pm\frac{i}{2}\big(\tr(\PhiXiMprime[2:p]\PhiXiM[2:p])\tr(\mlvec{H})\\
  &\qquad{} - \tr(\PhiXiM[2:p]\PhiXiMprime[2:p])\tr(\mlvec{H})\big)\\
  =& 0
\end{align}
\paragraph{Case 3: Either $\xi_1$ or $\xi_1^\prime=0$ and $\xi_1\neq \xi_1^\prime$.}
\begin{align}
  &(\ref{ln:app_expectation1})\\
  =&\EXP\<\Phi_{\mlvec{\xi}_{2:p}}(\mlvec{M}),0\circ\Phi_{\mlvec{\xi}_{2:p}^\prime}(\mlvec{M})\>\\
  =& 0
\end{align}
Combining Case 2 and Case 3, $\EXP\<\Phi_{\mlvec{\xi}}(\mlvec{M}),
\Phi_{\mlvec{\xi^\prime}}(\mlvec{M})\> = 0$ for $\mlvec{\xi}\neq
\mlvec{\xi}^\prime$. 

\paragraph{Second moments}
The correlation between the square of Frobenius norm can calculated recursively as:
\begin{align}
  & \EXP[\|\PhiXiM[1:p]\|_F^2\cdot\|\PhiXiMprime[1:p]\|_F^2]\\
  =& \EXP[\<\PhiXiM[1:p], \PhiXiM[1:p]\>\\
   &\qquad\qquad\quad\<\PhiXiMprime[1:p],\PhiXiMprime[1:p]\>]\\
 =& \EXP\<\PhiXiM[2:p] \Phimap{(0/1)}{l} \circ \PhiXiM[2:p]\>\\
  &\qquad\qquad\quad\cdot\<\PhiXiMprime[2:p] \Phimap{(0/1)}{l} \circ \PhiXiMprime[2:p]\>\\
  =& \frac{1}{4}\Big\{\EXP[\|\PhiXiM[2:p]\|_F^2\cdot\|\PhiXiMprime[2:p]\|_F^2]\\
\pm&\EXP[\tr(\PhiXiM[2:p]\mlvec{H}_l\PhiXiM[2:p]\mlvec{H}_l) \|\PhiXiMprime[2:p]\|_F^2]\\
\pm&\EXP[\|\PhiXiM[2:p]\|_F^2 \tr(\PhiXiMprime[2:p]\mlvec{H}_l\PhiXiMprime[2:p]\mlvec{H}_l)]\\
\pm&\EXP[
    \tr(\PhiXiM[2:p]\mlvec{H}_l\PhiXiM[2:p]\mlvec{H}_l) \\
  &\qquad\qquad\quad\cdot\tr(\PhiXiMprime[2:p]\mlvec{H}_l\PhiXiMprime[2:p]\mlvec{H}_l) 
    ]
    \big\}
  \\
  =&\frac{1 + O(d^{-2})}{4}
     \EXP[\|\PhiXiM[2:p]\|_F^2\cdot\|\PhiXiMprime[2:p]\|_F^2]
\end{align}

Therefore the diagonal elements of the Gram matrix has second moments:
\begin{align}
  & \EXP[\|\PhiXiM[1:p]\|_F^2\cdot\|\PhiXiM[1:p]\|_F^2]\\
  =&\frac{1 + O(d^{-2})}{4}\EXP[\|\PhiXiM[2:p]\|_F^2\cdot\|\PhiXiM[2:p]\|_F^2]\\
  =&(1 + O(pd^{-2}))\frac{\tr(\mlvec{M}^2)^2}{4^p}
\end{align}

We are now ready to calculate the second moments for the off-diagonal elements
of the Gram matrix. For $\mlvec{\xi}, \mlvec{\xi}\in\Xi$:
\begin{align}
  &\EXP\<\PhiXiM[l:p], \PhiXiMprime[l:p]\>^2\label{ln:app_square_of_product}\\
 =&\EXP\<\PhiXiM[l+1:p], (\Phi^{\xi_l}_l)^* \circ \Phi^{\xi_l^\prime}_l \circ\PhiXiMprime[l+1:p]\>^2
\end{align}

\paragraph{Case 1: $\xi_l = \xi_l^\prime$.}
\begin{align}
   &  (\ref{ln:app_square_of_product})\\
  = &\EXP\<\PhiXiM[l+1:p], \Phi^{(0/1)}_l \circ\PhiXiMprime[l+1:p]\>^2\\
  = &\frac{1}{4}\EXP\big[\tr(\PhiXiM[l+1:p] \PhiXiMprime[l+1:p])^2\\
   &\pm2\tr(\PhiXiM[l+1:p] \PhiXiMprime[l+1:p] )\\
   &\qquad\qquad\cdot\tr(\PhiXiM[l+1:p]\mlvec{H}_l\PhiXiMprime[l+1:p]\mlvec{H}_l)\\
   &+\tr(\PhiXiM[l+1:p]\mlvec{H}_l\PhiXiMprime[l+1:p]\mlvec{H}_l)^2\big]\\
  = &\frac{1}{4}\{(1\pm\frac{2}{d^2 - 1})\EXP\big[\tr(\PhiXiM[l+1:p] \PhiXiMprime[l+1:p])^2\big]\\
    & + O(d^{-2})\EXP[\|\PhiXiM[l+1:p]\|_F^2\cdot\|\PhiXiMprime[l+1:p]\|_F^2]\}
\end{align}
\paragraph{Case 2: $\xi_l,\xi_l^\prime\in\{1,2\}$ and $\xi_l\neq \xi_l^\prime$.}
Up to a sign flip, we have
\begin{align}
  &  (\ref{ln:app_square_of_product})\\
  = &\EXP\<\PhiXiM[l+1:p], \Phi^{(2)}_l \circ\PhiXiMprime[l+1:p]\>^2\\
  = &-\frac{1}{4}\EXP[\<\PhiXiM[l+1:p],\\
      &\qquad\quad\quad\mlvec{H}_l\PhiXiMprime[l+1:p] - \PhiXiMprime[l+1:p] \mlvec{H}_l\>^2]\\
  = &-\frac{1}{4}\EXP\{
      \tr(\PhiXiMprime[l+1:p] \PhiXiM[l+1:p]\mlvec{H}_l)^2\\
      &+\tr(\PhiXiM[l+1:p] \PhiXiMprime[l+1:p]\mlvec{H}_l)^2\\
  &-2\tr( \PhiXiM[l+1:p] \PhiXiMprime[l+1:p]\mlvec{H}_l)\\
  &\quad\cdot\tr( \PhiXiMprime[l+1:p] \PhiXiM[l+1:p]\mlvec{H}_l)
      \}\\
  = &-\frac{1}{2(d - d^{-1})}
      \EXP\big[\tr((\PhiXiM[l+1:p]\PhiXiMprime[l+1:p])^2)\\
      &-       \tr(\PhiXiM[l+1:p]^2 \PhiXiMprime[l+1:p]^2)\big]\\
  \leq & \frac{1}{d - d^{-1}}\EXP[
         \tr(\PhiXiM[l+1:p]^2) \tr(\PhiXiMprime[l+1:p]^2)]
\end{align}

\paragraph{Case 3: Either $\xi_l$ or $\xi_l^\prime=0$ and $\xi_l\neq
  \xi_l^\prime$.} For the case where $\xi_l \neq \xi_l^\prime$, and one of
$\xi_l$ and $\xi_l^\prime$ is $0$, the correlation 
$\EXP\<\PhiXiM[l:p], \PhiXiMprime[l:p]\>^2 = 0$.

Combining the above three cases, we have the variance bounded:
\begin{align}
  \mathbb{V}(\<\Phi_{\mlvec\xi}(\mlvec{M}), \Phi_{\mlvec\xi^\prime}(\mlvec{M})\>)
  =O(d^{-1})\frac{\tr(\mlvec{M}^2)^2}{4^p}
\end{align}
\end{proof}

\setcounter{equation}{0}
\section{Proofs for Upper Bounds}
\label{sec:app_upperbound}
For any $p$-parameter quantum circuit of consideration, we can express the
circuit as:
\begin{align}
  \mlvec{U}(\mlvec\theta) = \mlvec{V}_p(\theta_p)\mlvec{V}_{p-1}(\theta_{p-1})\cdots\mlvec{V}_1(\theta_1),
\end{align}
where $\mlvec{V}_l(\theta_l) = \exp(-i\theta_l\mlvec{H}_l)$ for some Hermitian
$\mlvec{H}_l$. We can bound the number of local minima in
$L(\mlvec\theta;\S)$ depending the choice of $\{\mlvec{H}_l\}_{l=1}^p$. In this section,
we provide a proof to Theorem~\ref{thm:upper} in
Section~\ref{sec:upperbound}:
\upperthm*
QNN with multi-qubit parameterized gates can have generating Hamiltonians $\{\mlvec{H}_l\}_{l=1}^p$
with more than two different eigenvalues. Specifically we consider Hamiltonians
with integral eigenvalues $\{E_1,\cdots,E_d\}$ such that
$\max_{c,c^\prime}|E_c-E_c^\prime| \leq \Delta$.
We can generalize Theorem~\ref{thm:upper} as:
\begin{theorem}[An upper bound for the more general setting]
  \label{thm:upper-gen}
  Consider $p$-parameter quantum neural networks composed of unitaries generated by
  Hamiltonians with integral spectrum gaps bounded by $\Delta$.
  The loss function $L(\mlvec\theta; \S)$ possesses at most
  $(4\Delta p)^p$ local minima, within each period, provided that the instance
  is not degenerated (i.e. the number of critical points is finite).
\end{theorem}
Note that any Hamiltonians with rational eigenvalues are included with proper scaling and shifting.

In Subection~\ref{subsec:app_ft}, we provide an upper bound on the Fourier
degree of the loss function. In Subsection~\ref{subsec:app_bezout}, we bound the
number of local minima for functions with bounded Fourier degree by considering
the number of roots of a polynomial system.
\subsection{Fourier Spectrum of the Loss Function}
\label{subsec:app_ft}
We first present a lemma on the Fourier spectrum of the loss function.
For all $l\in[p]$, let $\{E_i^{(l)}\}_{i=1}^d$ be the integral eigenvalues for $\mlvec{H}_l$ and let $\Delta_l$
denote the largest eigen-gap in absolute value: $\Delta_l:=
\max_{i,j\in[d]}|E_i^{(l)} - E_j^{(l)}|$. For arbitrary choice of training set $\S$, the
loss function $L(\mlvec\theta;\S)$ in Equation~\ref{eqn:loss} has the
following property:
\begin{lemma}[Fourier Transformation of the loss function: Generalized version]
\label{lm:FT}
  Let $\hat{L}: \real^p\rightarrow \complex$ be the
  Fourier Transformation of $L(\mlvec\theta;\S)$, namely for any
  $\mlvec{k}=(k_1,k_2,\cdots, k_p)^T\in\mathbb{Z}^p$, define:
  \begin{multline}
    \hat{L}(\mlvec k) := \frac{1}{T_1 T_2 \cdots T_p}\int_{[0,
      T_1]\times\cdots\times[0, T_p]} L(\mlvec\theta; \S)\\
    \cdot\exp\Bigl(-i \sum_{l=1}^p \frac{k_l \theta_l}{T_l}\Bigr)
    \mlvec d\mlvec{\theta}
  \end{multline}
  where $T_l$ is the period of $L(\mlvec\theta;\S)$ in $\theta_l$.
  Let $K$ be the support of $\hat{L}$ (i.e. $K:=\{\mlvec k\in\real^p |
  \hat{L}(\mlvec k) \neq 0\}$).
  The Fourier degree of the loss function $\Delta_K := \max_{\mlvec
    k\in K}\sum_{l=1}^p|k_l|$ is bounded by $\sum^p_{l=1} \frac{T_l\cdot\Delta_{l}}{\pi}$.
\end{lemma}
\begin{proof}
  For all $l\in[p]$, let $\{\mlvec{u}_i^{(l)}\}_{i=1}^d$ be the eigenvectors of $\mlvec{H}_l$ with corresponding
  eigenvalues $\{E_i^{(l)}\}_{i=1}^d$:
  \begin{equation}
  \mlvec{H}_l = \sum_{i=1}^dE_i^{(l)} \mlvec{u}_i^{(l)}
  {\mlvec{u}_i^{(l)}}^\dagger
  \end{equation}
  The unitary gate paramtrized by $\theta_l$ is therefore
  \begin{align}
  \mlvec{V}_l(\theta_l) = \exp(-i\theta_l\mlvec{H}_l) = \sum_{i=1}^de^{-i\theta_lE_i^{(l)}}\mlvec{u}_i^{(l)}{\mlvec{u}_i^{(l)}}^\dagger
  \end{align}
  And the unitary gate $\mlvec{U}(\mlvec{\theta})$ can be written as:
  \begin{align}
    \mlvec{U}(\mlvec{\theta}) &= \mlvec{V}_p(\theta_p)\mlvec{V}_{p-1}(\theta_{p-1})\cdots \mlvec{V}_1(\theta_1)\\
                              &= \bigl(\sum_{i_p\in[d]}e^{-i\theta_p E^{(p)}_{i_p}} {\mlvec{u}_{i_p}^{(p)}}{\mlvec{u}_{i_p}^{(p)}}^{\dagger}\bigr)\\
                                &\qquad\cdots\bigl(\sum_{i_1\in[d]}e^{-i\theta_1 E^{(1)}_{i_1}} {\mlvec{u}_{i_1}^{(1)}}{\mlvec{u}_{i_1}^{(1)}}^{\dagger} \bigr)\\
                              &=\sum_{\mlvec{i}\in[d]^p}
                                \Bigl(\prod_{l=1}^{p-1}{\mlvec{u}_{i_{l+1}}^{(l+1)}}^\dagger {\mlvec{u}_{i_{l}}^{(l)}}\Bigr)
                                \Bigl( \prod_{l=1}^pe^{-i\sum_{l=1}^p\theta_l E_{i_l}^{(l)}} \Bigr)\\
                                &\qquad\cdot(\mlvec{u}_{i_p}^{(p)}{\mlvec{u}_{i_1}^{(1)}}^\dagger)
  \end{align}
  The output of the neural network given the density matrix $\mlvec{\rho}$ is therefore
  \begin{align}
    f(\mlvec\rho,\mlvec\theta)
    &= \tr\Bigl(\mlvec{V}_p(\theta_p)\mlvec{V}_{p-1}(\theta_{p-1})\cdots \mlvec{V}_1(\theta_1)\mlvec\rho\\
    &\qquad\cdot\mlvec{V}_1(\theta_1)^\dagger \mlvec{V}_{2}(\theta_{2})^\dagger \cdots \mlvec{V}_p(\theta_p)^\dagger\mlvec{M} \Bigr)\\
    &=\sum_{\mlvec{i}\in[d]^p}\sum_{\mlvec{j}\in[d]^p}
      \hat{f}_{\mlvec{ij}}(\mlvec\rho) \cdot e^{i\sum_{l=1}^p\theta_l(E^{(l)}_{j_l}-E^{(l)}_{i_l})}
    \end{align}
    where for any $\mlvec{i}, \mlvec{j} \in [d]^p$
    \begin{multline}
      \hat{f}_{\mlvec{ij}}(\mlvec\rho)=
      \bigl({\mlvec{u}_{i_1}^{(1)}}^\dagger\mlvec{\rho}{\mlvec{u}_{j_1}^{(1)}}\bigr)
      \bigl({\mlvec{u}_{j_p}^{(p)}}^\dagger\mlvec{M}{\mlvec{u}_{i_p}^{(p)}}\bigr)\\
      \cdot\Bigl(\prod_{l=1}^{p-1}{\mlvec{u}_{i_{l+1}}^{(l+1)}}^\dagger{\mlvec{u}_{i_{l}}^{(l)}}{\mlvec{u}_{j_l}^{(l)}}^\dagger{\mlvec{u}_{j_{l+1}}^{(l+1)}}\Bigr)
  \end{multline}
  This indicates the Fourier coefficients of $f(\mlvec{\rho};\mlvec{\theta})$ is
  supported on a subset of $\tilde{K}:=\{(k_1,\cdots,k_p) |\forall l\in[p],
  \exists i,j\in[d]: k_l =
  \frac{(E^{(l)}_{i}-E^{(l)}_{j})T_l}{2\pi}\}$, and that the Fourier
  degree of $f(\mlvec\rho,\mlvec\theta)$ is bounded by
  $\sum_{l=1}^p\frac{T_l\cdot\Delta_l}{2\pi}$.
  Therefore for arbitray $\mlvec{\rho}$ and label
  $y$, the Fourier degree of the square loss $\bigl(
  f(\mlvec{\rho},\mlvec\theta)-y \bigr)^2$, $\Delta_K\leq
  \sum_{l=1}^p\frac{T_l\cdot\Delta_l}{\pi}$.
  Same holds for loss function $L(\mlvec\theta;\S)$
  with arbitrary training set $\S$.
\end{proof}
  For $\mlvec{H}_l$ with integral eigenvalues, 
  \begin{align}
     & \exp(-i(\theta_l + 2\pi)\mlvec{H}_l)\\
    =& \sum_{i=1}^de^{-i(\theta_l + 2\pi)E_i^{(l)}}\mlvec{u}_i^{(l)}{\mlvec{u}_i^{(l)}}^\dagger\\
    =& \exp(-i\theta_l\mlvec{H}_l).
  \end{align}
  Hence $T_l\leq 2\pi$ and $\Delta_{K} 
  \leq 2\sum_{l=1}^p\Delta_l$. Let $\Delta$ be $\max_l\Delta_l$, the Fourier
  degree $\Delta_K$ of the loss function is bounded by $2\Delta p$. 

For Hamiltonians with two-levels, we have the following corollary:
\begin{corollary}
  For quantum neural network instances composed of unitaries generated by
  two-level Hamiltonians, the Fourier degree of the loss function
  $L(\mlvec\theta;\S)$ is bounded by $2p$ for arbitrary dataset $\S$.
\end{corollary}
\begin{proof}
  As shown earlier, for any Hermitian $\mlvec{M}$ and for $\mlvec H_l$ with the eigenvalues $\pm
  1$, 
  the output $f(\mlvec{\rho},\mlvec{\theta})$ is periodic in $\pi$ for each coordinate.
  Also notice for all $l\in[p]$, $\Delta_l = 2$. Hence the Fourier degree
  $\Delta_K$ of
  $L(\mlvec\theta;\S)$ is bounded by $2p$.
\end{proof}

\subsection{Change of Variable and Root Counting}
\label{subsec:app_bezout}
In this subsection, we elaborate on the change of variable and the upper bound on
number of critical points by B\'ezout's Theorem. This would complete our proof for
Theorem~\ref{thm:upper} and \ref{thm:upper-gen}.

Let $T_l$ be the period of $L(\mlvec{\theta}; \S)$ in $\theta_l$, and 
$\hat{L}(\mlvec k)$ the Fourier coefficient for $\mlvec k = (k_1, \cdots, k_p)^T\in \mathbb{Z}^p$. We have
\begin{multline}
  L(\mlvec\theta;\S)
             = \sum_{\mlvec k\in K}\hat{L}(\mlvec k)\Bigl(\cos\frac{k_1\theta_1}{T_1}+ i\sin\frac{k_1\theta_1}{T_1}\Bigr)\cdot\\
               \cdots \cdot \Bigl(\cos\frac{k_p\theta_p}{T_p}+ i\sin\frac{k_p\theta_p}{T_p}\Bigr)
\end{multline}
Here $K \subseteq \mathbb{Z}^p$ is the support of the Fourier coefficients.

By definition, a local minimum must be a critical point, hence it suffices to bound
the number of critical points for $L(\mlvec{\theta};\S)$.
Define $G_l(\mlvec\theta)$ as
\begin{align}
  G_l(\mlvec\theta) &=\frac{\partial}{\partial\theta_l}L(\mlvec\theta;\S)\\
  &= \sum_{\mlvec k\in K}k_l\hat{L}(\mlvec k)\bigl(-\sin\frac{k_l\theta_1}{T_l}+ i\cos\frac{k_l\theta_l}{T_l}\bigr)\\
  &\cdot\prod_{l^\prime\neq l} \bigl(\cos\frac{k_{l^\prime}\theta_{l^\prime}}{T_{l^\prime}}+ i\sin\frac{k_{l^\prime}\theta_{l^\prime}}{T_{l^\prime}}\bigr)
\end{align}
We can tell from above expression that the Fourier spectrum of $G_l$ is supported on the same set $K$.
A critical point of $L(\mlvec{\theta};\S)$ must satisfies that for all
$l\in[p]$, $G_l(\mlvec\theta) = 0$.

By induction, $\cos k\theta$ can be expressed as a degree-$k$ polynomial of
$\cos\theta$ and $\sin k\theta$ as a degree-$(k-1)$ polynomial of $\cos\theta$ multiplied by $\sin\theta$. Consider the change of variable
\begin{align}
  c_l = \cos (\theta_l/T_l),\ s_l = \sin({\theta_l}/{T_l}),\
  \forall l\in[p].
\end{align}
Let $g_l(c_1, s_1, \cdots, c_p, s_p)$ be the multivariate polynomial constraints
corresponding to $G_l(\mlvec\theta)$ after the change of variable:
\begin{multline}
    \sum_{\mlvec k\in K}k_l\hat{L}(\mlvec k)
    \bigl(-s_lU_{k_l-1}(c_l) + iT_{k_l}(c_l)\bigr)\\
    \cdot\prod_{l^\prime\neq l}
    \bigl(T_{k_{l^\prime}}(c_{l^\prime})+ is_{l^\prime} U_{k_{l^\prime}-1}(c_{l^\prime})\bigr)
\end{multline}
where $T_k(\cdot)$ and $U_k(\cdot)$ are Chebyshev polynomials of the first and
second kind.
For each
$g_l$, the sum of degrees of $c_{l^\prime}$ and $s_{l^\prime}$ is bounded by
$\max_{\mlvec{k}\in K}|k_{l^\prime}|$, and the degree $\mathsf{deg}(g_l)$ of
$g_l$ is bounded by $\Delta_K = \max_{\mlvec{k}\in}\sum_{l=1}^p|k_l|$.
The change of variable is one-to-one from $\theta_l\in[0,T_l)$ to a pair of $(c_l,
s_l)\in\real^2$ under the constraint $c_l^2 + s_l^2 = 1$. Therefore, it suffices
to count the number of roots of the polynomial system with $2p$ parameters
and $2p$ constraints: 
\begin{align}
  \begin{cases}
    &g_1(c_1, s_1, \cdots, c_p, s_p)  = 0,\\ 
    &\vdots\\
    &g_p(c_1, s_1, \cdots, c_p, s_p)  = 0,\\ 
    &h_1(c_1, s_1, \cdots, c_p, s_p) = c_1^2 + s_1^2 - 1 =0,\\
    &\vdots\\
    &h_p(c_1, s_1, \cdots, c_p, s_p) = c_p^2 + s_p^2 - 1 =0.\\
  \end{cases}
\end{align}

Notice that for general polynomial system, the number of critical points can be
unbounded. For example, consider a system composed of constant polynomials,
every point in the domain is a critical point. This corresponds to constant
loss function, where the gradients vanishes everywhere with positive
semidefinite Hessians. For this reason, we will focus on the non-degenerated case
with finitely many critical points. Under the premise of non-degeneracy, 
B\'ezout's Theorem (e.g. Section 3.3 in \cite{cox2006using}) states that the number of roots can bounded by
the product of degree of polynomial constraints
$2^p\mathsf{deg}(g_1)\mathsf{deg}(g_2)\cdots\mathsf{deg}(g_p)$.
$ \leq (2\Delta_K)^p$.

We also prove a similar result for the more general case where the generators are
Hamiltonians with integral spectrum: let $\Delta$ be the largest eigen-gap for
each of the Hamiltonians, the number of local minima within each period is
upper bounded by $O((\Delta p)^p)$. 
Combined with results in Subsection~\ref{subsec:app_ft}, the proof for Theorem~\ref{thm:upper} and
\ref{thm:upper-gen} is complete.

\setcounter{equation}{0}
\section{Numerical Results}
\label{sec:app_num}
For all the experiments in this section, we
study the $p$-parameter QNN as mentioned in
Example~\ref{example:base} in \supref{sec:app_constructions}, where :
\begin{align}
  \mlvec{M} &:= \otimes_{l=1}^p(\mlvec{Y} + \mlvec{I})\\
  \mlvec H_l &:= \bigl( \otimes_{r=1}^{l-1}\mlvec{I} \bigr) \otimes \mlvec{Z} \otimes \bigl( \otimes_{l+1}^p\mlvec{I} \bigr),\ \forall l\in[p]
\end{align}

In \supref{subsec:app_exp_optimization}, we provide details and
more numerical results for experiments described in
Section~\ref{sec:experiments}.
In \supref{subsec:app_exp_visual}, we visualize the $2$-d loss
landscape of Example~\ref{example:address_decomposability}.

\subsection{Training with Gradient-based Methods}
\label{subsec:app_exp_optimization}
In this subsection, we use Example~\ref{example:address_minima_position} to
demonstrate that our construction can be hard to train with \textit{gradient-based methods}.
The loss function of the example can be expressed as
\begin{multline}
  \frac{1}{2p}\sum_{l=1}^p\big((\sin2\theta_l - \sin\frac{\pi}{50})^2\\ + \frac{1}{4}(\cos2\theta_l - \cos\frac{\pi}{50})^2 \big).
\end{multline}
with global minima at
$\mlvec\theta^\star =  (\frac{\pi}{100},\cdots,\frac{\pi}{100})^T$.

\paragraph{Hyperparameters.} For all three optimizers,
we choose the (initial) learning rate to be $0.01$. For \textsf{RMSProp}, we
choose the smoothing constant $\alpha$ for mean-square estimation to be $0.99$.
For \textsf{Adam}, we set the averaging coefficients $\beta_1=0.9$ for the
gradients and $\beta_2=0.999$ for the its square. For \textsf{L-BFGS} we choose
the history size to be $100$. The numbers of iterations for training are set to
$200$-th for each pair of instances and optimizers; as can be seen in
Figure~\ref{fig:training-curve}, all pairs have already converged at the
$100$-th iteration.

\paragraph{Training curves} We plot the
training curve for QNN instances with $2$, $4$, $8$, $16$ and $32$ parameters
with \textsf{Adam}, \textsf{RMSProp} and \textsf{L-BFGS}.
For each pair of instance and optimizer, we repeat the
experiments with uniform random initialization. As shown in
Figure~\ref{fig:training-curve},
for all the experiment setting considered here, while all initialization
converge efficiently, there are initializations that does not converge to the
global minima.
\begin{figure}[!htbp]
  \centering
  \subfigure[2 parameters: \textsf{Adam}]{
    \includegraphics[width=.3\linewidth]{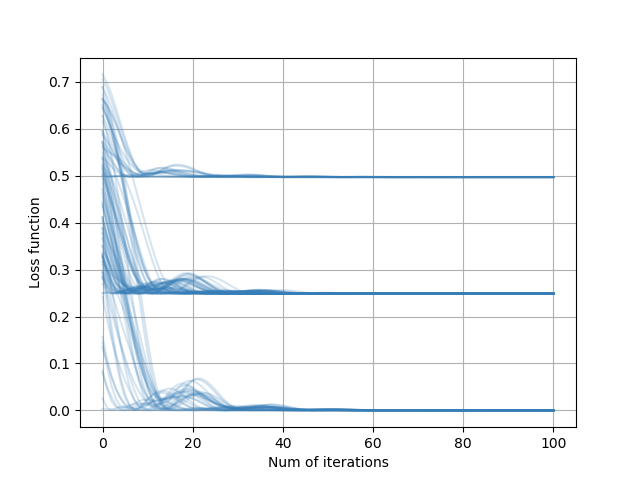}
  }
  \subfigure[2 parameters: \textsf{RMSProp}]{
    \includegraphics[width=.3\linewidth]{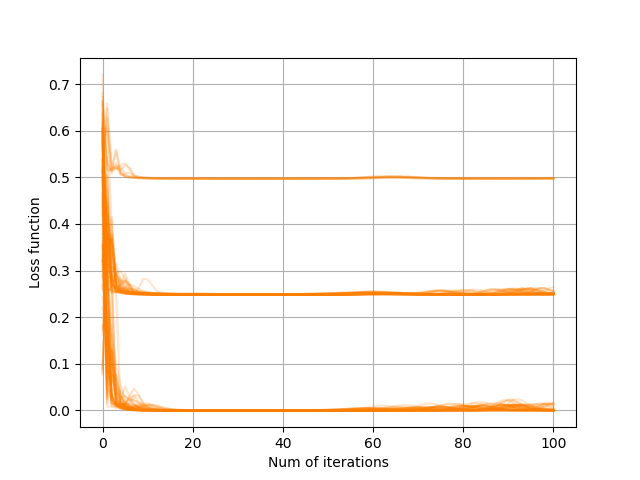}
  }
  \subfigure[2 parameters: \textsf{L-BFGS}]{
    \includegraphics[width=.3\linewidth]{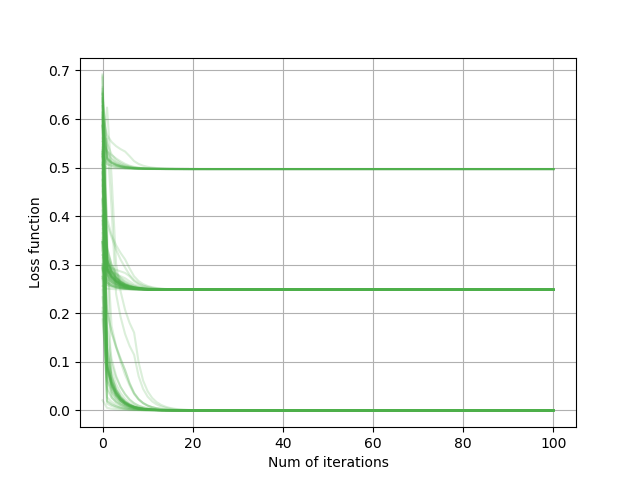}
  }\\
  \subfigure[4 parameters: \textsf{Adam}]{
    \includegraphics[width=.3\linewidth]{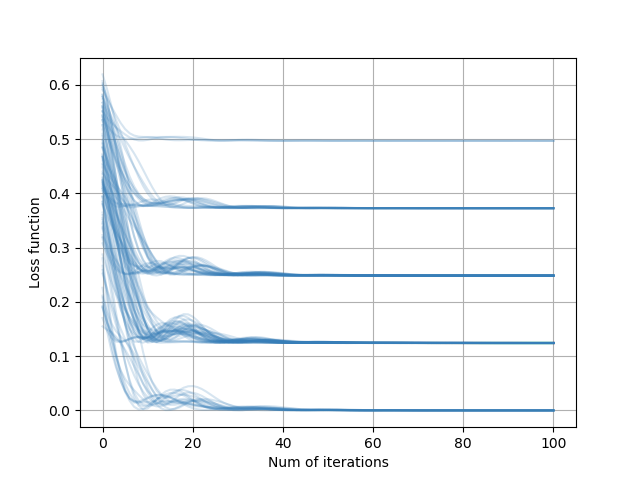}
  }
  \subfigure[4 parameters: \textsf{RMSProp}]{
    \includegraphics[width=.3\linewidth]{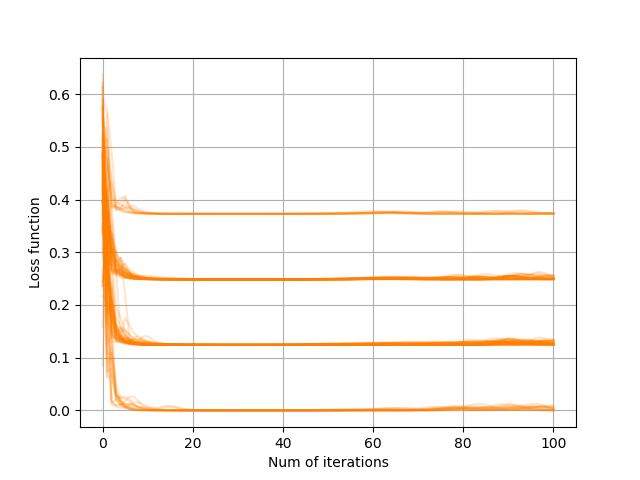}
  }
  \subfigure[4 parameters: \textsf{L-BFGS}]{
    \includegraphics[width=.3\linewidth]{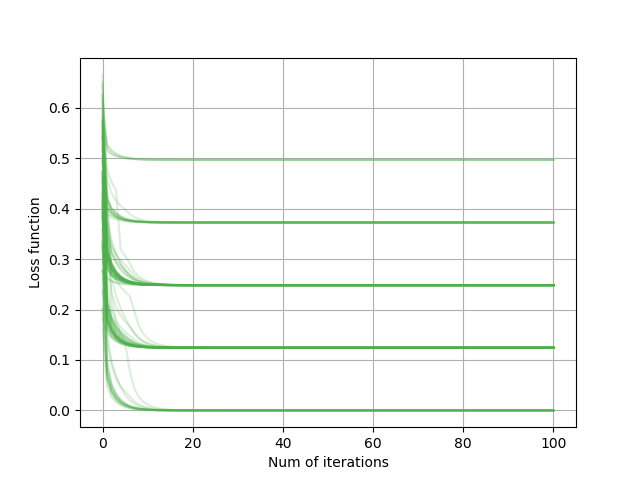}
  }\\
  \subfigure[8 parameters: \textsf{Adam}]{
    \includegraphics[width=.3\linewidth]{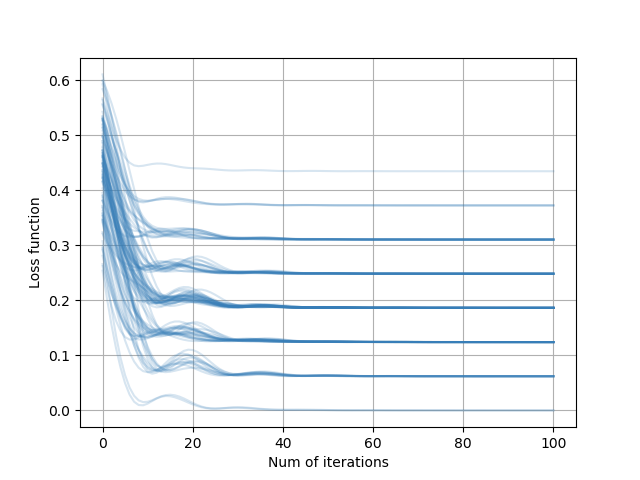}
  }
  \subfigure[8 parameters: \textsf{RMSProp}]{
    \includegraphics[width=.3\linewidth]{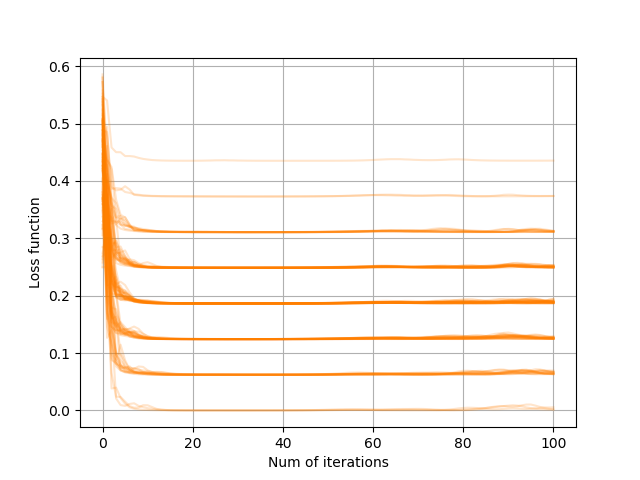}
  }
  \subfigure[8 parameters: \textsf{L-BFGS}]{
    \includegraphics[width=.3\linewidth]{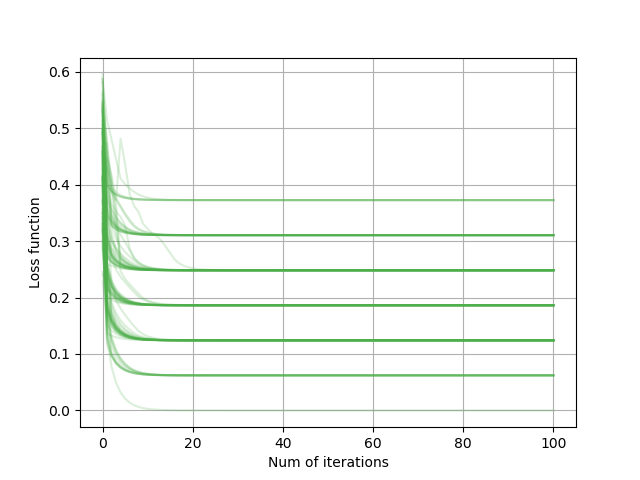}
  }\\
  \subfigure[16 parameters: \textsf{Adam}]{
    \includegraphics[width=.3\linewidth]{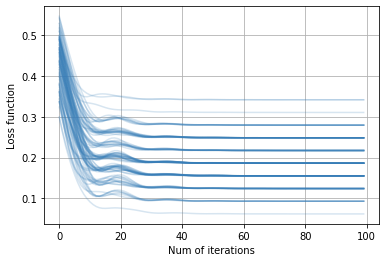}
  }
  \subfigure[16 parameters: \textsf{RMSProp}]{
    \includegraphics[width=.3\linewidth]{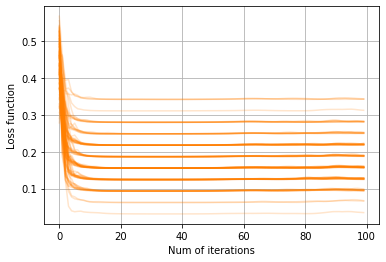}
  }
  \subfigure[16 parameters: \textsf{L-BFGS}]{
    \includegraphics[width=.3\linewidth]{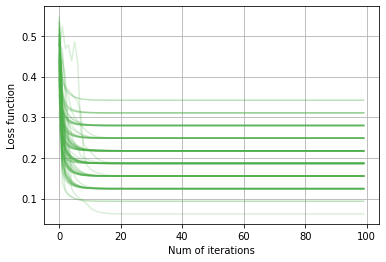}
  }
  \\
  \subfigure[32 parameters: \textsf{Adam}]{
    \includegraphics[width=.3\linewidth]{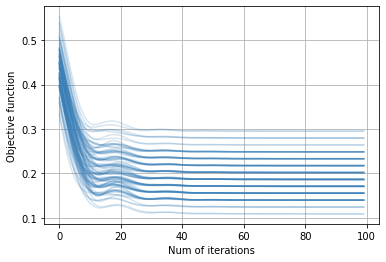}
  }
  \subfigure[32 parameters: \textsf{RMSProp}]{
    \includegraphics[width=.3\linewidth]{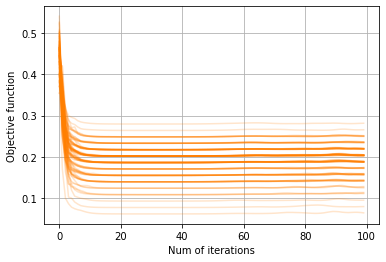}
  }
  \subfigure[32 parameters: \textsf{L-BFGS}]{
    \includegraphics[width=.3\linewidth]{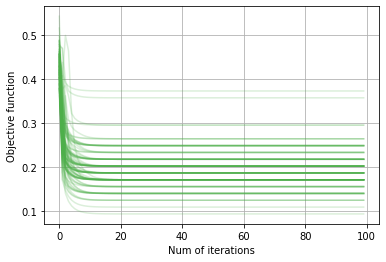}
  }
  \caption{
    Empirical risk minimization of different QNN instance with \textsf{Adam}, \textsf{RMSProp} and
    \textsf{L-BFGS}.
    For each experiment setting we repeat 100 times.
    }
  \label{fig:training-curve}
\end{figure}

\paragraph{More on distributions of function values} In
Figure~\ref{fig:more-dist}, we plot more results on the distribution of function
values under \textsf{RMSProp} and have similar observation as mentioned in Section~\ref{sec:experiments}.
\begin{figure}[!htbp]
  \centering
  \subfigure[2 parameters: \textsf{RMSProp}]{
    \includegraphics[width=.9\linewidth]{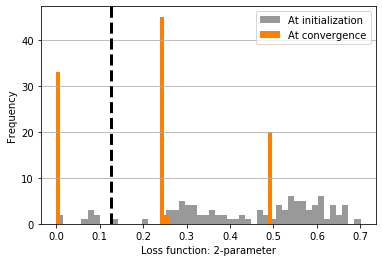}
  }\\
  \subfigure[4 parameters: \textsf{RMSProp}]{
    \includegraphics[width=.9\linewidth]{\imghome/pennylane-dist-4.png}
  }\\
  \subfigure[6 parameters: \textsf{RMSProp}]{
    \includegraphics[width=.9\linewidth]{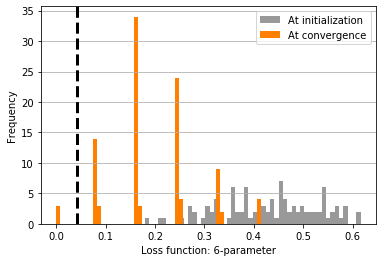}
  }\\
  \subfigure[8 parameters: \textsf{RMSProp}]{
    \includegraphics[width=.9\linewidth]{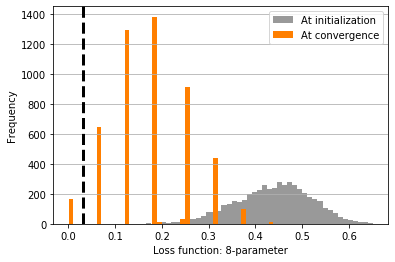}
  }
\end{figure}
\begin{figure}[!htbp]
  \subfigure[10 parameters: \textsf{RMSProp}]{
    \includegraphics[width=.9\linewidth]{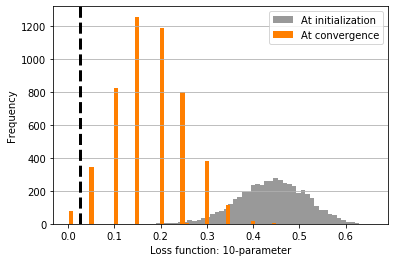}
  }\\
  \subfigure[12 parameters: \textsf{RMSProp}]{
    \includegraphics[width=.9\linewidth]{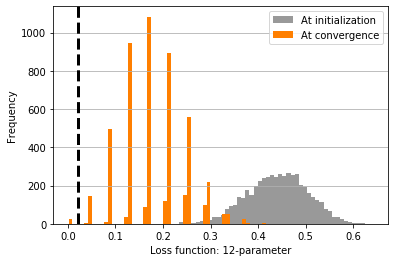}
  }\\
  \subfigure[14 parameters: \textsf{RMSProp}]{
    \includegraphics[width=.9\linewidth]{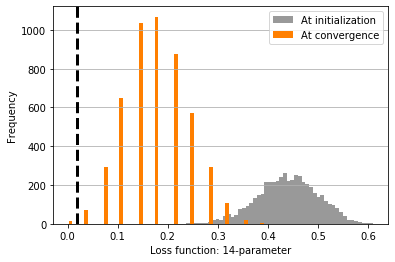}
  }\\
  \subfigure[16 parameters: \textsf{RMSProp}]{
    \includegraphics[width=.9\linewidth]{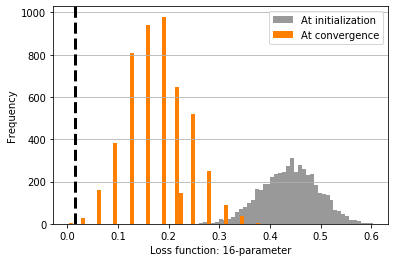}
  }
  \end{figure}
\begin{figure}[!htbp]
  \subfigure[18 parameters: \textsf{RMSProp}]{
    \includegraphics[width=.9\linewidth]{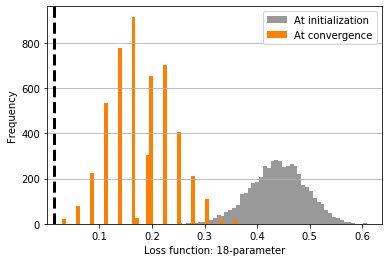}
  }
  \caption{
    Distribution of function values of QNN instances with \textsf{RMSProp}. For instances
  with size $2,4,6$, the experiments are repeated 200 times; for the rest of the
instances, the experiments are repeated 5000 times.}
  \label{fig:more-dist}
\end{figure}

\paragraph{Calculating the empirical probabilities} For all the instances of
consideration, the function values of local minima can be calculated. For
$p$-parameter instances, the function value of global minima is $0$, and for the
other local minima, the function values are at least $0.5/p$. For calculating
the empirical probability that random initialization converges to the global
minima, we count the number of trial that converge to values less than $0.25/p$.

\subsection{Visualization: Non-decomposable Construction}
\label{subsec:app_exp_visual}
For low-dimensional cases, it is possible to visualize the loss function of Example~\ref{example:address_decomposability} of the
construction by plotting the contour of the landscape. In
Figure~\ref{fig:contour}, we plot the contour of our construction for
$p=2$, with loss function proportional to
\begin{multline}
  (\sin 2\theta_1 - \sin\frac{\pi}{50})^2 + (\sin 2\theta_2 -
  \sin\frac{\pi}{50})^2\\
  +  \frac{1}{16} ((\cos 2\theta_1 - \cos\frac{\pi}{50} )^2
  + \frac{3}{2}(\cos 2\theta_2 - \cos\frac{\pi}{50})^2) \\
  + \frac{1}{8} (\cos 2\theta_1\cos 2\theta_2 - \cos^2\frac{\pi}{50})^2
\end{multline}
The global minima are $(k_1\pi + \frac{\pi}{100}, k_2\pi + \frac{\pi}{100})$ with $k_1,k_2\in\integer$.
Within each period, there are a total of $4$ local minima where black box local
search methods might stuck at. Among them, the global minima are marked in
black. The gradient-based methods only converge to the global minimum when the
initial value of the parameter lies in certain region.
\begin{figure}[!htbp]   
  \begin{center}
    \includegraphics[width=\linewidth]{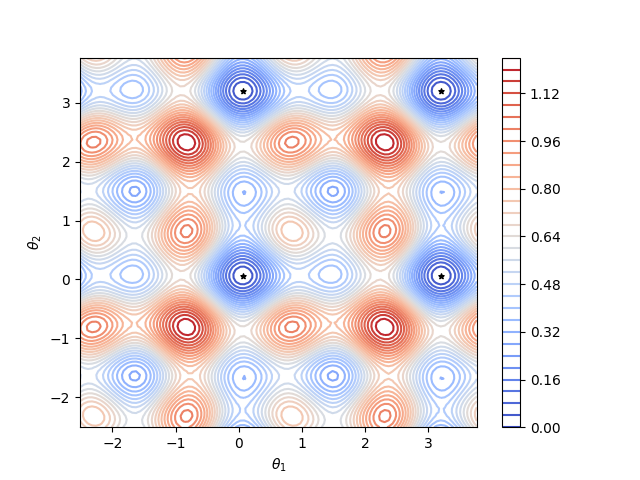}
    \caption{Landscape of the constructed QNN instance with $2$ qubits and $2$
      parameters. The global optima are marked in black. }
    \label{fig:contour}
  \end{center}
\end{figure}

\subsection{Robustness of the Constructions}
\label{subsec:app_noise}
Our construction above demonstrates that in the worst-case, under-parameterized
QNNs can have exponentially many local minima. It is natural to ask whether the
local minima in our constructions are stable under perturbation.
To this end, we repeated our experiments with Gaussian noises $\mathcal{N}(0, \sigma)$ added to the
labels. The function values at local minima, as shown in
Figure~\ref{fig:noisy_dist} (Cf. Figure~\ref{fig:dist0},\ref{fig:dist} in the main text and Figure~\ref{fig:more-dist}(h) 
in the supplementary material), have
changed, as the noise breaks the symmetry of sub-optimal minima. But as shown in Figure~\ref{fig:noisy_decay} (Cf.
Figure~\ref{fig:decay} in the main text), the
exponential decay of success rate in finding the global minima remains
for different $\sigma$ (recall that the labels in our construction are in
$[0,1]$). We used \textsf{RMSProp} optimizer, with other hyperparameters the
same as the pervious experiments. 

Moreover, by direct calculation of the suboptimality gaps and eigenvalues of Hessians at
local minima, it can be proved that our examples are indeed robust against random
label perturbations, quantum noise due to noisy gates, or due to the finite number of
measurements, and even \textbf{adversarial} perturbations, as long as the
resulting perturbation in the loss function is bounded in $\ell_\infty$-norm.

\begin{figure}[!htbp]
    \centering
    \subfigure[list][Function values at initialization and at convergence for the
    16-parameter instance with noisy labels, repeated for 5000 random
    initializations.]{
      \includegraphics[width=\linewidth]{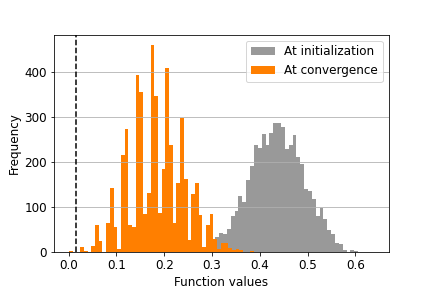}
      \label{fig:noisy_dist}
    }
    \subfigure[List entry][The exponential decay of success rate for finding the global minimum under 10000
    random initialization with label noise $\mathcal{N}(0,\sigma)$.]{
      \includegraphics[width=\linewidth]{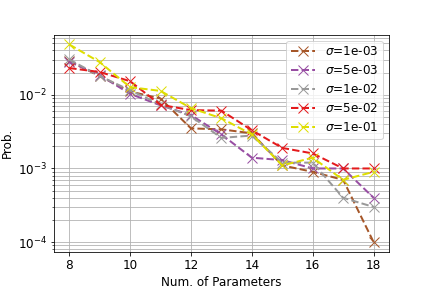}
      \label{fig:noisy_decay}
    }
    \caption{Empirical risk minimization with noisy labels. (a) the function values
      at convergence for a 16-parameter instance; the perturbation breaks the
      symmetry of the local minima, hence the more continuous spectrum of
      function values (Cf. Figure~\ref{fig:dist}(h)). (b) the decay of
      success rate for finding the global minima; the exponential tendency
      remains in the presence of Gaussian label noise up to $\sigma=1e-1$.}
  \end{figure}

\subsection{More Experimentrs on Datasets beyond Our Construction}
\label{subsec:app_common_datasets}
So far we have observed exponentially many local minima in the datasets in our
construction. Now we turn to more natural datasets that may appear in practice. Specifically,
we consider the following family of datasets with a clear interpretation as an encoding of a
classical, linearly separable concept: for the $p$-parameter instance,
we first randomly choose $\mlvec{w} \in \real^{2p}$ as the normal vector to the
separating hyperplane. The classical dataset $\{(\mlvec{x},
y)|\mlvec{x}\in\real^{2p}, y\in\{0,1\}\}$ is generated as follows:
(1) uniformly draw the feature vector $\mlvec{x} = (x_1, \cdots, x_p, x_{p+1}, \cdots, x_{2p})^T$
from $[0, 2\pi]^{2p}$; (2) $y = 1$ if $\mlvec{w}^T \mlvec{x} > 0$ and $y = 0$
otherwise. The classical feature $\mlvec{x}$ is encoded into a quantum state
$\mlvec{\rho}(\mlvec{x}) = \ket{\Psi(\mlvec{x})}\bra{\Psi(\mlvec{x})}$ using the two-layer XY-encoder:
$\Psi(\mlvec{x}) :=
\otimes_{l=1}^p\exp(-i x_{p+l} \mlvec{Y}_l)
\otimes_{l=1}^p\exp(-i x_{l} \mlvec{X}_l)
\ket{0}^{\otimes p}$. This process is repeated to construct a 100-sample
dataset.
For each QNN instance, we sampled 70 initial points and optimize the 
mean-square loss with \textsf{RMSProp} for 2000 iterations. The rest of the
settings are the same as our original experiments.
In Figure~\ref{fig:common} (Cf. Figure~\ref{fig:dist} in the main text and Figure~\ref{fig:more-dist}~(a)-(d)
in the supplementary material), we trained instances with 2,4,6,8 qubits, each with 70
random initialization, and plotted the distribution of function values at
convergence. There is a large number of local minima, and only a few
random initialization ended up at the global minima.
While we no longer have a clear exponential
dependency, we did  observe that as the number of parameters increases, the number
of local minima increases significantly, and the success rate for finding global minima drops sharply. 
Such a phenomenon is also resilient to random choices of $\mlvec{w}$ and random sampling of feature vectors. 
This could be initial numerical evidence supporting the generality of our
observed phenomena.

\begin{figure}[!htbp]
  \centering
   \subfigure[2 parameters: \textsf{RMSProp}]{
     \includegraphics[width=.8\linewidth]{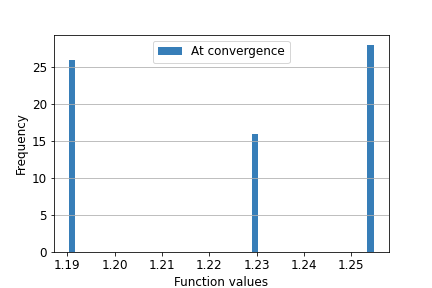}
   }
   \subfigure[4 parameters: \textsf{RMSProp}]{
     \includegraphics[width=.8\linewidth]{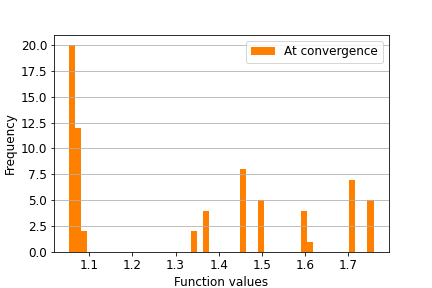}
   }
 \end{figure}
 \begin{figure}[!htbp]
   \subfigure[6 parameters: \textsf{RMSProp}]{
     \includegraphics[width=.8\linewidth]{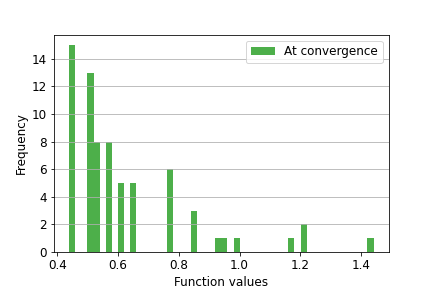}
   }\\
     \subfigure[8 parameters: \textsf{RMSProp}]{
     \includegraphics[width=.8\linewidth]{\imghome/common_dist_8.png}
   }
  \caption{
    Empirical risk minimization for the common dataset using \textsf{RMSProp}.
    For each experiment setting, we repeat for 70 random initializations and run
    for 2000 iterations. The number of local minima increases significantly with the number of
    parameters.
  }
  \label{fig:common}
\end{figure}

\end{document}